\documentclass[journal,onecolumn,12pt,draftclsnofoot]{IEEEtran}
\usepackage[letterpaper, left=0.8in, right=0.8in, bottom=0.9in, top=0.75in]{geometry}
\usepackage{amsmath}
\usepackage{amsfonts}
\usepackage{amsthm}
\usepackage{cite}
\usepackage{soul,color}
\usepackage{graphicx}
\usepackage{float}
\usepackage[font=footnotesize]{caption}
\usepackage[font=footnotesize]{subcaption}
\usepackage{epstopdf}
\usepackage{amssymb}
\usepackage{stmaryrd}
\usepackage{multirow}
\newtheorem{theorem}{\textit{Theorem}}
 
\usepackage{enumitem}
\usepackage{algorithm}
\usepackage{algorithmic}

\begin{document}
\title{\huge Bidirectional Optical Spatial Modulation for Mobile Users: Towards a Practical Design for LiFi Systems}

\author{Mohammad Dehghani Soltani,
        Mohamed Amine Arfaoui,
        Iman Tavakkolnia$^*$,
        Ali Ghrayeb,
        Majid Safari,
        Chadi Assi, 
        Mazen Hasna, and
        Harald Haas
\thanks{M. D. Soltani, I. Tavakkolnia, M. Safari, and H. Haas are with the LiFi Research and Development Centre, Institute for Digital Communications, School of Engineering, University of Edinburgh, UK. e-mail: \{m.dehghani, i.tavakkolnia, majid.safari, h.haas\}@ed.ac.uk. M. A. Arfaoui and C. Assi are with Concordia Institute for Information Systems Engineering (CIISE), Concordia University, Montreal, Canada, e-mail:\{m\_arfaou@encs, assi@ciise\}.concordia.ca. Ali Ghrayeb is with Texas A \& M University at Qatar, Doha, Qatar, e-mail: ali.ghrayeb@qatar.tamu.edu. Mazen Hasna is with Qatar University, Doha, Qatar, e-mail: hasna@qu.edu.qa. $^*$Corresponding author: Iman Tavakkolnia, i.tavakkolnia@ed.ac.uk}}

\maketitle

\vspace{-1cm}

\begin{abstract}
Among the challenges of realizing the full potential of light-fidelity (LiFi) cellular networks are user mobility, random device orientation and blockage. In the paper, we study the impact of those challenges on the performance of LiFi networks in an indoor environment using measurement-based channel models, unlike existing studies that rely on theoretical channel models. In our work, we adopt spatial modulation (SM), which has been shown to be energy efficient in many applications, including LiFi. We consider two configurations for placing the photodiodes (PDs) on the user equipment (UE). The first one is referred to as the screen receiver (SR) whereby all the PDs are located on one face of the UE, e.g., the screen of a smartphone, whereas the other one is a multi-directional receiver (MDR), in which the PDs are located on different sides of the UE. The latter configuration was motivated by the fact that SR exhibited poor performance in the presence of random device orientation and blockage. In fact, we show that MDR outperforms SR by over $10$ dB at bit-error ratio (BER) of $3.8\times10^{-3}$. Moreover, an adaptive access point (AP) selection scheme for SM is considered where the number of APs are chosen adaptively in an effort to achieve the lowest energy requirement for a target BER and spectral efficiency. The user performance with random orientation and blockage in the whole room is evaluated for sitting and walking activities. For the latter, we invoke the orientation-based random waypoint (ORWP) mobility model. We also study the performance of the underlying system on the uplink channel where we apply the same techniques used for the downlink channel. Specifically, as the transmitted uplink power is constrained, the energy efficiency of SM is evaluated analytically. It is shown that the multi-directional transmitter (MDT) with adaptive SM is highly energy efficient. Furthermore, as a benchmark, we compare the performance of the proposed framework to that of the conventional spatial multiplexing system, and demonstrate the superiority of the proposed one.   
\end{abstract}

\section{Introduction}
\subsection{Motivation}
It is anticipated that the mobile data traffic will generate about $49$ exabyte per month and the average global mobile connection speed will surpass $20$ Mbps by $2021$ \cite{Cisco}. The total number of smartphones (including phablets) will be over $50\%$ of global devices and they will generate around $86\%$ of the mobile data traffic ($42$ exabyte per month) by $2021$. Therefore, both academia and industry are looking for alternative solutions to offload heavy traffic loads from radio frequency (RF) wireless networks. 
Light-Fidelity (LiFi) is a novel bidirectional, high-speed and fully networked optical wireless communication (OWC) system which can be employed as a complementary structure along with RF networks\cite{Haas}. LiFi utilizes visible light and infrared spectra in downlink and uplink, respectively, and provides high data rates in short distances \cite{Bian,tsonev2015towards}. 
Compared to RF networks, LiFi offers notable benefits such as providing enhanced security, utilizing a very large and unregulated bandwidth and energy efficiency. These advantages have put LiFi in the scope of recent and future research. A task group for LiFi in IEEE $802.11$ already exists
\cite{ieee802bb} as well as a task group for optical camera communication (OCC) in IEEE $802.15.7{\rm r}1$ \cite{ieee802task}.
\subsection{Related Work}
Different data transmission techniques, which originate from RF wireless communication, are modified and adopted for OWC. Many of the developed communication techniques are validated by experiments and are even standardized \cite{Haas,ieee802}. For instance, the single-carrier modulation format on-off keying (OOK) is  used in IEEE 802.15.7 \cite{ieee802} as a simple technique which could also provide dimming. However, high data rates cannot be achieved using OOK. Therefore, parallel transmission techniques are proposed to increase the spectral efficiency. Multi-carrier modulation techniques, such as orthogonal frequency division multiplexing (OFDM) \cite{mossaad2015visible}, wavelength division multiplexing (WDM) \cite{chun2016led}, and multiple-input-multiple-output (MIMO) techniques \cite{fath2013performance} are among the most common realizations of parallel data transmission. Each of these techniques has several variants which make them favorable in various conditions \cite{tavakkolnia2018energy}. 

Spatial modulation (SM) is a type of MIMO structure that offers enhanced spectral efficiency compared to non-MIMO systems, and is more energy efficient with lower complexity as compared to full MIMO (i.e., spatial multiplexing while using all available transmitters) \cite{MeslehHaasSM}. In SM, part of the information is mapped on the degrees of freedom in the spatial domain, and the remaining part is mapped on the signal domain \cite{di2014spatial}. In optical SM, the selection of one or more of light emitting diodes (LEDs) forms the \textit{spatial information}. A modulation format is also used to map the \textit{signal information} for each selection of LEDs. Usually, pulse amplitude modulation (PAM) is used for modulating the signal information \cite{mesleh2011optical}. Space shift keying (SSK) and generalized space shift keying (GSSK) are two modulation schemes that are used when OOK is chosen for signal modulation \cite{jeganathan2009space,mesleh2008spatial,jeganathan2008generalized}. The performance of SM is well studied theoretically and experimentally, and its advantages and potential practical applications are highlighted \cite{popoola2014demonstration}. An important benefit of SM is the absence of interference from other transmitter units in single user scenarios. However, as expected for any MIMO system, the performance of SM also heavily depends on the channel conditions \cite{mesleh2008spatial}. This problem is more severe in OWC where the transmitter and/or receiver units are usually placed close to each other, and thus, the channel matrix may become ill-conditioned \cite{fath2013performance,mesleh2011optical}. Moreover, the performance of the system significantly changes for different environments, activities, and user positions.

Device orientation can significantly affect the users' throughput. Most of the studies on OWC assume that the device always faces vertically upward. This assumption may have been driven by the lack of having a proper model for orientation, and/or to make the analysis tractable. Nonetheless, such an assumption is only accurate for a limited number of devices (e.g., laptops with a LiFi dongle). However, the majority of users use devices such as smartphones, and in real-life scenarios, users are mobile and tend to hold their device in a way that feel most comfortable. This means that the device is not always facing upward and thus can have any orientation. Only a few studies have considered the impact of random orientation in their analysis, see for instance \cite{ICCTilting,MDSArxiv2018Orientation,BEROrientation} and references therein. All these works signify the importance of incorporating device orientation. The other important metric that can influence the system performance is the blockage of the optical channel by the user itself, known as self-blockage, or by other users or objects, and this consequently can interrupt the communication link. Blockage has been modeled in both millimeter wave and LiFi systems \cite{BlockageGlob03,blockageIET,raghavan2018statistical}.\vspace{-0.3cm}

\subsection{Contributions}
Against the above background, we present in this paper a bidirectional communication framework for an indoor LiFi environment. We adopt downlink and uplink channel models derived from real-life measurements, which makes the proposed framework relevant to the deployment efforts of LiFi networks. The adopted models encompass the combined effect of user mobility, random orientation and blockage. We note that this is the first time that such factors are incorporated into the design and analysis of LiFi networks when SM is employed. Among the performance measures that we consider are the bit-error ratio (BER), the spectral efficiency and energy efficiency. Motivated by the fact that SM is highly energy efficient, we adopt a variation of it to ensure that the least amount of energy is needed to achieve a target BER and spectral efficiency. \\
\indent It has been now well established that correlation among the channel gains between the transmit and receive antennas affects the performance of SM. To this end, we consider two representative configurations of PD placements. In the first configuration, we assume that the PDs are placed uniformly at one end of the UE, and in the second configuration, the PDs are placed on different edges of the UE. We refer to the first configuration as screen receiver/transmitter (SR/ST), and multi-directional receiver/transmitter (MDR/MDT) for the second configuration. The motivation behind introducing the latter configuration is twofold. First, SM performs best when the channel gains are uncorrelated, i.e., the performance degrades with correlation. Second, the random orientation may give rise to the problem of ill-conditioned channel matrixes for the SR/ST configuration, suggesting that a few subchannels become inadequate to support reliable transmission.    

The random orientation of the UE is modeled based on the experimental measurements reported in \cite{MDSArxiv2018Orientation} to obtain the instantaneous orientation of PDs. The measurement data is collected according to the normal daily activities of a number of participants for sitting and walking activities. Furthermore, not only is line-of-sight (LOS) considered but also non-line-of-sight (NLOS) channel gains are included to have an accurate channel model. Also, the blockage of the optical channel by human users and other random objects is considered. Therefore, a channel model close to realistic scenarios is incorporated which makes the methods and results presented in this paper reliable for future system design. 

We study the robustness of MDR/MDT in conjunction with SM against random orientation and blockage and show the impact of the channel on the overall performance. To improve the performance further, we propose using an adaptive SM (ASM) scheme for downlink and uplink in which the order of SM (i.e., the number of active light sources) is determined based on the strength of the channels between the LEDs and PDs. We examine the proposed ASM scheme for both sitting and walking activities. For sitting activities, about $10^4$ locations in a $5$ m $\times$ $5$ m room are considered, while for the walking activities, an orientation-based random waypoint (ORWP) mobility model is applied. It is observed that such adaptive methods significantly improve the performance. As a benchmark, we compare the proposed adaptive method with a spatial multiplexing MIMO system for different spectral efficiencies. The results confirm that the proposed ASM is more efficient. Moreover, the proposed MDR method can achieve up to twice the spectral efficiency of SR for the same SNR and target BER. 

To support bidirectional OWC, the validity of the MDR/MDT with ASM is also evaluated for the uplink channel. However, due to the constraint on the transmission uplink power, it is important to analyze the achievable energy efficiency in the uplink. This is performed by deriving bounds on the achievable spectral efficiency and defining an energy efficiency criterion. It is indicated that ASM, when integrated with the MDT, can provide higher energy efficiency compared to that of conventional SM methods.

In light of the above discussion, we may summarize the paper contributions as follows.
\begin{itemize}
\item Based on real-life measurements, we adopt a practical channel model that incorporates LOS and NLOS channel gain components, user mobility, user UE random orientation and link blockage. The effect of each of these phenomena on the performance of SM is studied.
\item We investigate the impact of different component placement configurations on the system performance. We show that placing the Tx/Rx LEDs/PDs on different sides of the UE (i.e., MDT/MDR) makes the system robust against blockage and random orientation.
\item We propose an ASM scheme in an effort to optimize the system performance. We propose algorithms for selecting the optimal number of used light sources for a given target spectral efficiency and reliability.
\item The performance of both downlink and uplink is investigated over the whole area of a typical indoor environment for walking and sitting activities. It is demonstrated that the MDR/MDT structure along with ASM improves the performance significantly. We also show that the proposed framework is superior to spatial multiplexing MIMO systems.
\item We analytically derive an energy efficiency metric which can be used for both uplink and downlink. We show that the proposed system is highly energy efficient for uplink when the MDT structure is combined with ASM. 
\end{itemize}\vspace{-0.3cm}


\subsection{Outline and Notations} 
The rest of the paper is organized as follows. The system model is presented in Section \ref{sec_system}. Section \ref{Sec_rand_ori_blk} presents the random orientation and link blockage modeling. In Sections \ref{sec_downlink} and \ref{sec_uplink}, the downlink and uplink performance is studied, respectively. Finally, the paper is concluded in Section \ref{sec_conc}, and future research directions are highlighted.

The following notations are adopted throughout the paper. Upper case bold characters denote matrices and lower case bold characters denote vectors. The set of natural numbers is denoted by $\mathbb{N}$ and the set of $N$-dimensional real-valued numbers is denoted by $\mathbb{R}^N$. The discrete set $\left\{1,2,...,N \right\}$ is denoted by $\llbracket 1,N \rrbracket$. Matrix transposition is denoted by the superscript $\{ \cdot \}^\mathrm{T}$. $|| \cdot ||_2$ denotes the Euclidean norm. $\mathcal{N}(\textbf{0}_N, \textbf{G})$ denotes the $N$-dimensional multivariate Gaussian probability distribution with zero-mean and covariance matrix $\textbf{G}$. The expected value is denoted by $\mathbb{E}[\cdot]$, the differential entropy is denoted by $h( \cdot )$ and the mutual information by $I(\cdot;\cdot)$. Superscript $C^+$ denotes max$(C,0)$. 

\section{System Model}\label{sec_system}
In this section, the channel model is described, and basics of SM are explained. Both downlink and uplink are studied in this paper and the system model can be used for both of them. Differences between downlink and uplink are highlighted wherever it is essential throughout the paper.
\vspace{-0.2cm}
\subsection{Channel Model}
The intensity modulation direct detection (IM/DD) optical wireless MIMO channel is considered, where $N_\mathrm{t}$ light sources (e.g. one or more LEDs) can transmit the signal and one UE receives the signal with $N_\mathrm{r}$ PDs. The resulting channel is described as:\vspace{-0.3cm}
\begin{equation}
\mathbf{y}=\mathbf {H}\mathbf{x}+\mathbf{n},
\end{equation}
where $\mathbf{x}$ is the transmitted signal vector of size $N_\mathrm{t}\times 1$; and $\mathbf{y}$ and $\mathbf{n}$ are $N_\mathrm{r}\times1$ vectors respectively representing the received signal and noise at each PD. The noise here includes all possible noises, such as shot noise and thermal noise and is assumed to be real valued additive white Gaussian $\mathcal{N} \left(\textbf{0}_{N_\mathrm{r}}, \sigma_{\rm n}^2 \textbf{I}_{N_\mathrm{r}} \right)$ and independent of the transmitted signal \cite{fath2013performance}. The variance of the noise is equal to $\sigma_{\rm n}^2=N_0B$, where $N_0$ is the single sided power spectral density of noise and $B$ is the bandwidth. The channel matrix $\mathbf{H}$ is given by: \vspace{-0.2cm}
\begin{equation}
\mathbf {H}= 
 \begin{pmatrix}
  h_{1,1} &  \cdots & h_{1,N_\mathrm{t}} \\
  \vdots    & \ddots & \vdots  \\
  h_{N_\mathrm{r},1} & \cdots & h_{N_\mathrm{r},N_\mathrm{t}} 
 \end{pmatrix},
\end{equation}
where the entities $h_{i,j}$ ($i=1, ...,N_\mathrm{r}$ and $j=1, ...,N_\mathrm{t}$) are the channel gain of the link between the $j$th transmitter and the $i$th PD which can be expressed as:\vspace{-0.2cm}
\begin{equation}\label{h}
h_{i,j}=h_{i,j}^\mathrm{LOS}+h_{i,j}^\mathrm{NLOS},
\end{equation}

\begin{figure}[!t]
\centering
\begin{subfigure}[b]{0.5\columnwidth}
\centering
\includegraphics{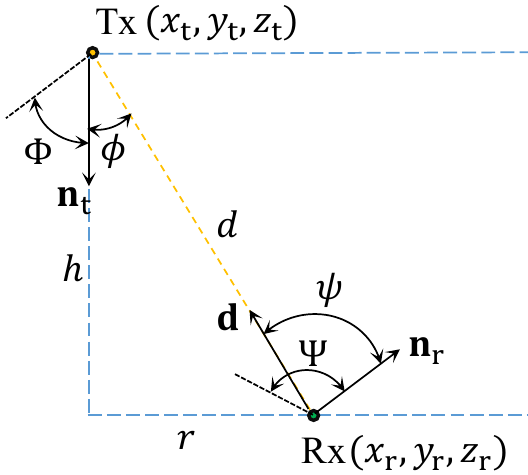}
\caption{LOS}
\label{LOS}
\vspace{-0.3cm}
\end{subfigure}~
\begin{subfigure}[b]{0.5\columnwidth}
\centering
\includegraphics{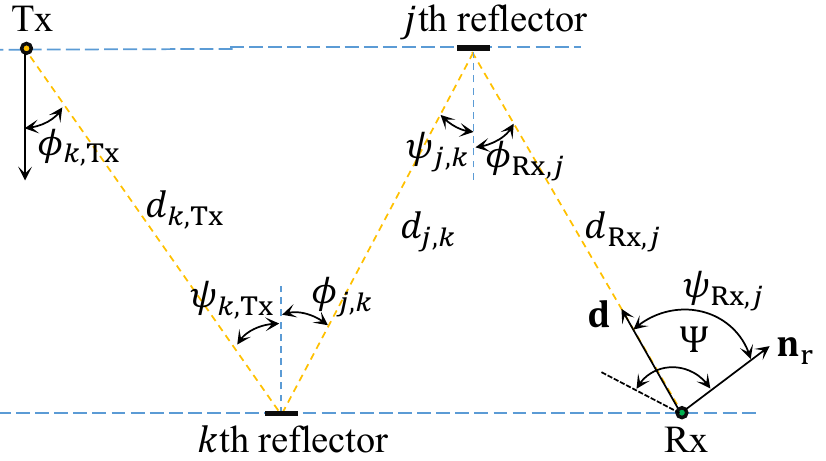}
\caption{NLOS}
\label{NLOS}
\vspace{-0.3cm}
\end{subfigure}
\caption{The downlink geometry of optical wireless communications with randomly-orientated user device.}
\label{Diag1}
\vspace{-0.9cm}
\end{figure}

\noindent where $h_{i,j}^\mathrm{LOS}$ is the LOS and $h_{i,j}^\mathrm{NLOS}$ is the NLOS channel gain. The channel gain depends on the distance between transmitter and receiver pairs (i.e., user position) and the orientation of each PD. Fig.~\ref{LOS} shows the LOS link geometry for a pair of receiver (Rx) and transmitter (Tx), where $\phi$ is the angle of radiance. The LED half-power semiangle is denoted by $\Phi$, and $\psi$ is the incidence angle. The receiver field-of-view (FOV) is shown by $\Psi$. The LOS channel gain of an optical link between a light source and a PD is given by \cite{kahn1997wireless}:
\begin{equation}\label{h_LOS}
h_{i,j}^\mathrm{LOS}=\frac{k+1}{2\pi d^2}A \cos^k(\phi)\cos(\psi)\mathrm{rect}\left(\frac{\psi}{\Psi}\right),
\end{equation}
where $A$ is the PD area, and $k=-1/\log_2(\cos(\Phi))$ is the Lambertian emission order. Furthermore, $\mathrm{rect}(\frac{\psi}{\Psi})=1$ for $0\leq\psi\leq\Psi$ and $0$ otherwise. 

The NLOS component of the channel gain can be calculated based on the method described in \cite{NLOSSchulze}. Using the frequency domain instead of the time domain analysis, we are able to consider an infinite number of reflections to have an accurate value of the diffuse link. The environment is segmented into a number of surface elements which reflect the light beams. These surface elements are modeled as Lambertian radiators described by \eqref{h_LOS} with $k=1$. 
Then, the NLOS channel gain including an infinite number of reflections can be expressed by \cite{NLOSSchulze}:\vspace{-0.1cm}
\begin{equation}\label{h_NLOS}
h_{i,j}^\mathrm{NLOS}=\mathbf{r}^\mathrm{T}\mathbf{G}_\rho(\mathbf{I-EG_\rho})^{-1}\mathbf{t},
\end{equation} 
where vectors $\mathbf{t}$ and $\mathbf{r}$ respectively represent the LOS link between the transmitter Tx and all the surface elements of the room and from all the surface elements of the room to the receiver Rx. Also, $(.)^\mathrm{T}$ denotes the transpose operator. Matrix $\mathbf{G}_\rho={\rm{diag}}(\rho_1,...,\rho_N)$ is the reflectivity matrix of all $N$ reflectors; $\mathbf{E}$ is the LOS transfer functions of size $N\times N$ for the link between all surface elements, and $\mathbf{I}$ is the unity matrix. In \eqref{h_NLOS}, the elements of $\mathbf{E}$, $\mathbf{r}$ and $\mathbf{t}$ are found according to \eqref{h_LOS} and Fig. \ref{Diag1} between pairs of Rx, Tx, and surface elements. In this paper, we assume that the modulation bandwidth is within the 3 dB bandwidth of the optical wireless transmission channel. Therefore, temporal delay between different Tx-Rx pairs is negligible, the temporal dispersion can be neglected, and only the DC channel gain is considered including LOS and NLOS components \cite{mesleh2011optical}.


The performance of a MIMO system depends heavily on the channel matrix. Devices such as laptops are usually placed on a flat surface and the PDs can be assumed to retain their orientation during each communication session \cite{MDSFeedback}, whether upward or not. On the other hand, hand-held devices such as smartphones, are prone to random changes in orientation due to hand motion. In this study, we focus on these types of devices and include the random orientation in our analysis. Moreover, objects and people may be placed close to the UE and block all or part of the light to reach one or more PDs. The details of random orientation and blockage modeling will be presented in Section \ref{Sec_rand_ori_blk}.\vspace{-0.3cm}

\subsection{Spatial Modulation} \label{sec_sm}

SM was first introduced in \cite{MeslehHaasSM}, which can provide the spectral and energy efficiency fulfillment of the next generation wireless communications. We review the basics of SM in this section and elaborate on how we adapt SM to the optical communication. More details can be found in \cite{DirenzoHanzoSM,mesleh2008spatial,mesleh2011optical,jeganathan2009space,imanGWC,MeslehRenzoSM,YounisRenzoSM} and references therein.

Following the basic principles of SM, the spatially distributed light sources are utilized to carry data along with the transmitted signal. In the original SM format \cite{mesleh2011optical}, only one light source is turned on at each time instant. Let $N_\mathrm{a}\leq N_\mathrm{t}$ be the number used LEDs, chosen out of $N_\mathrm{t}$ LEDs. Thus, by activating only one LED at each channel use, $\log_2(N_\mathrm{a})$ bits (the spatial information) are transmitted by SM. The transmitted symbol by an individual LED is also encoded by an $M$-ary PAM ($M$-PAM) constellation. Therefore, the spectral efficiency is $R=\log_2 (M) + \log_2(N_\mathrm{a})$ bits/sec/Hz. Note that, unlike spatial multiplexing (i.e., full MIMO), even one PD can be sufficient for signal detection because only differences between all possible symbols determine the system performance. This highlights the benefit of SM which is simple and is capable of potentially satisfying communication requirements when some PDs are blocked and not available. 

In this paper, we consider activating one of the available $N_\mathrm{a}$ LEDs with an $M$-PAM modulation format, which results in a total of $K=M N_\mathrm{a}$ symbols. The intensity levels of $M$-PAM are given by:\vspace{-0.0cm}
\begin{equation}
I_m=\frac{2I}{M+1}m,~~~~\mathrm{for}~m=1, \dots, M,
\end{equation}
where $I$ is the average emitted optical power. Therefore, one of the available $N_\mathrm{a}$ LEDs transmits one of the $M$ levels at each channel use, and the input vector $\mathbf{x}=\mathbf{s}_k$, $k=1, ..., K$, is chosen from the columns of the $N_a\times K$ matrix $\mathbf{S}=[I_1 \mathbf{I}_{N_a}~I_2 \mathbf{I}_{N_a}~\cdots I_M \mathbf{I}_{N_a}]$, where $\mathbf{I}_{N_a}$ is the square unity matrix of size $N_\mathrm{a}$. At the Rx, Maximum-Likelihood (ML) detection is performed. An error occurs whenever a transmitted vector $\mathbf{s}_{1}$ is detected mistakenly as another vector $\mathbf{s}_{2}$. The pairwise error probability (PEP) is defined as:\vspace{-0.0cm}
\begin{equation}\label{pep}
\mathrm{PEP}=Q\left(\sqrt[]{\frac{\gamma_\mathrm{Tx}}{4 I^2}\|\mathbf{H}(\mathbf{s}_1-\mathbf{s}_2)\|^2}\right),
\end{equation}
where $\gamma_\mathrm{Tx}$ is the average transmit SNR, and $\mathrm{Q}(\cdot)$ is the Q-function. The transmit SNR is defined as $\gamma_\mathrm{Tx}=\frac{E_s}{N_0}$, where $E_s$ is the mean emitted electrical energy. We define the received SNR $\gamma_\mathrm{Rx}$ by considering the received signal energy as the total received signal energies at all $N_\mathrm{r}$ PDs. Therefore, the received SNR can be expressed as \cite{fath2013performance}:\vspace{-0.0cm}
\begin{equation}
\gamma_\mathrm{Rx}=\frac{\gamma_\mathrm{Tx}}{N_\mathrm{a}^2}\sum_{i=1}^{N_\mathrm{r}}\left(\sum_{j=1}^{N_\mathrm{a}}h_{i,j}\right)^2.
\end{equation}
The upper bound on the BER can be derived using the union bound method as:
\begin{equation}\label{bersm}
\mathrm{BER}\left(M, E_s, \mathbf{H} \right) \simeq \frac{1}{K\log_2(K)} \sum\limits_{k_1=1}^{K}\sum\limits_{k_2=1}^{K} \mathrm{d_{H}}(b_1,b_2) Q\left(\sqrt{\frac{\gamma_\mathrm{Tx}}{4I^2}||\mathbf{H}(\mathbf{s}_1-\mathbf{s}_2)||^2}\right),
\end{equation}
where $\mathrm{d_{H}}(b_1,b_2)$ is the Hamming distance between the two bit allocations of $b_1$ and $b_2$ corresponding to signal vectors $\mathbf{s}_1$ and $\mathbf{s}_2$. It has been shown in the literature \cite{fath2013performance} and later in the paper that \eqref{bersm} is a tight bound at high SNR. 

It can be seen from \eqref{pep} and \eqref{bersm} that the error performance of SM directly depends on the channel matrix which determines the differentiability between signal vectors. 
We can assume that all PDs are placed on the screen of a smartphone, as shown in Fig. \ref{smartphone1}. However, this results in poor performance due to two issues. First, the resulting channel matrix is likely to be highly ill-conditioned because PDs are placed close to each other and this gives rise to correlation. Second, it is highly likely that some of the transmitters are out of the FOV of all PDs since usually the smart phone is held with an orientation other than upward. Therefore, we propose another structure by placing PDs on the screen and three other sides as shown in Fig. \ref{smartphone2}. Note that another PD can be placed at the back that can be activated instead of the one on the screen for situations where the user is lying on a horizontal surface. We call this structure the ``multi-directional receiver'' (MDR), which solves both above-mentioned problems. We investigate the performance of both structures later in the paper. We refer to the structure in Fig. \ref{smartphone1} as the screen receiver (SR). It should be noted that, in either structure, since the PDs are located at the top of the cellphone, it is a very low probability that they will be covered by the user's hand when the cellphone is being used. 

\begin{figure}[t]
		\centering
		\begin{subfigure}[b]{0.4\columnwidth}
			\centering
			\includegraphics[width=0.8\columnwidth,draft=false]{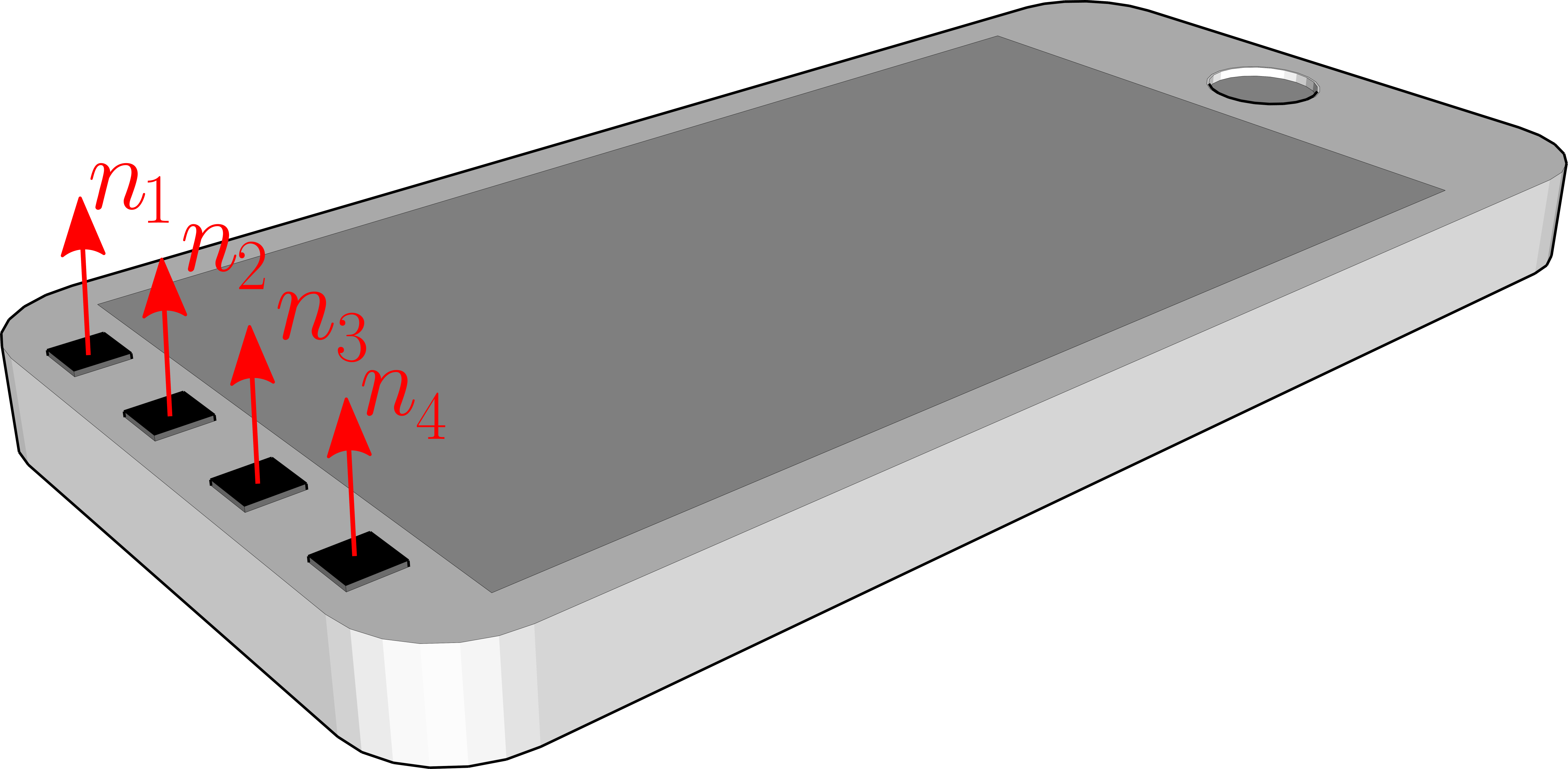}
						\caption{}
				\label{smartphone1}
		\end{subfigure}%
		~
		\begin{subfigure}[b]{0.4\columnwidth}
			\centering
			\includegraphics[width=0.8\columnwidth,draft=false]{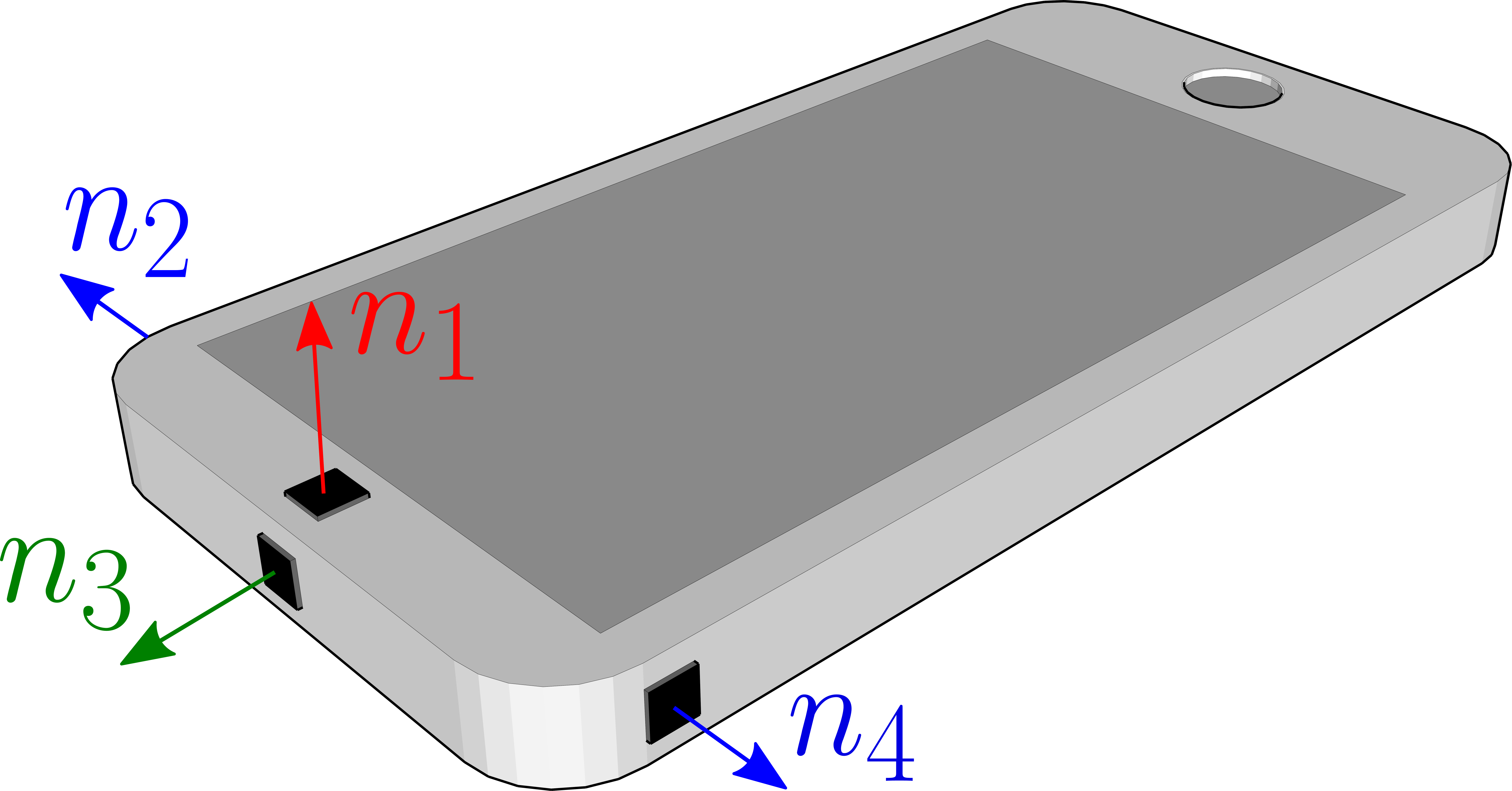}
						\caption{}
				\label{smartphone2}
		\end{subfigure}
        \vspace{-0.3cm}
				\caption{The SR and MDR structures for hand-held smartphone}
				\label{phones}
                \vspace{-0.6cm}
\end{figure}
\section{Random Orientation and Link Blockage} \label{Sec_rand_ori_blk}
In this section, we first present the statistics of device orientation based on experimental results. Then, the link blockage either due to the user itself or other objects in the vicinity will be explained. We use the models presented in this section along with the channel model described in the previous section to calculate the channel matrix for any user location and selection of Tx/Rx.
\vspace{-0.3cm}

\subsection{Random Orientation Modeling}
Current smartphones are equipped with gyroscope, accelerometer and compass that enable them to obtain the orientation in three dimensions by measuring the elemental rotation angles yaw, $\alpha$, pitch, $\beta$, and roll, $\gamma$ \cite{MDSHandover}. As shown in Fig.~\ref{figori}, $\alpha$, $\beta$ and $\gamma$ denote rotations about $z$-axis, $x$-axis and $y$-axis, respectively.  
According to the Euler's rotation theorem, any rotation matrix can be expressed by 
$\mathbf{R}=\mathbf{R}_\alpha \mathbf{R}_\beta \mathbf{R}_\gamma,$
where
\begin{equation}
\mathbf{R}_\alpha=\begin{bmatrix}
		\cos \alpha & -\sin \alpha & 0 \\
		\sin \alpha & \cos \alpha & 0 \\
		0 & 0 & 1 
\end{bmatrix},~~~
\mathbf{R}_\beta=\begin{bmatrix}
		1 & 0 & 0 \\
		0 & \cos \beta & -\sin \beta \\
		0 & \sin \beta & \cos \beta 
		\end{bmatrix},~~~
\mathbf{R}_\gamma=\begin{bmatrix}
		\cos \gamma & 0 & \sin \gamma \\
		0 & 1 & 0 \\
		-\sin \gamma & 0 & \cos \gamma
				\end{bmatrix},
\end{equation}
and angles $\alpha$, $\beta$, and $\gamma$ are shown in Fig. \ref{figori}. The normal vector of a PD can be described by $\mathbf{n}_\mathrm{r}=\mathbf{R}\mathbf{n}_0$, where $\mathbf{n}_0$ is the orientation vector for the vertically upward case. 

	\begin{figure}[t]
		\centering
		\begin{subfigure}[b]{0.5\columnwidth}
			\centering
			\includegraphics[width=0.5\columnwidth,draft=false]{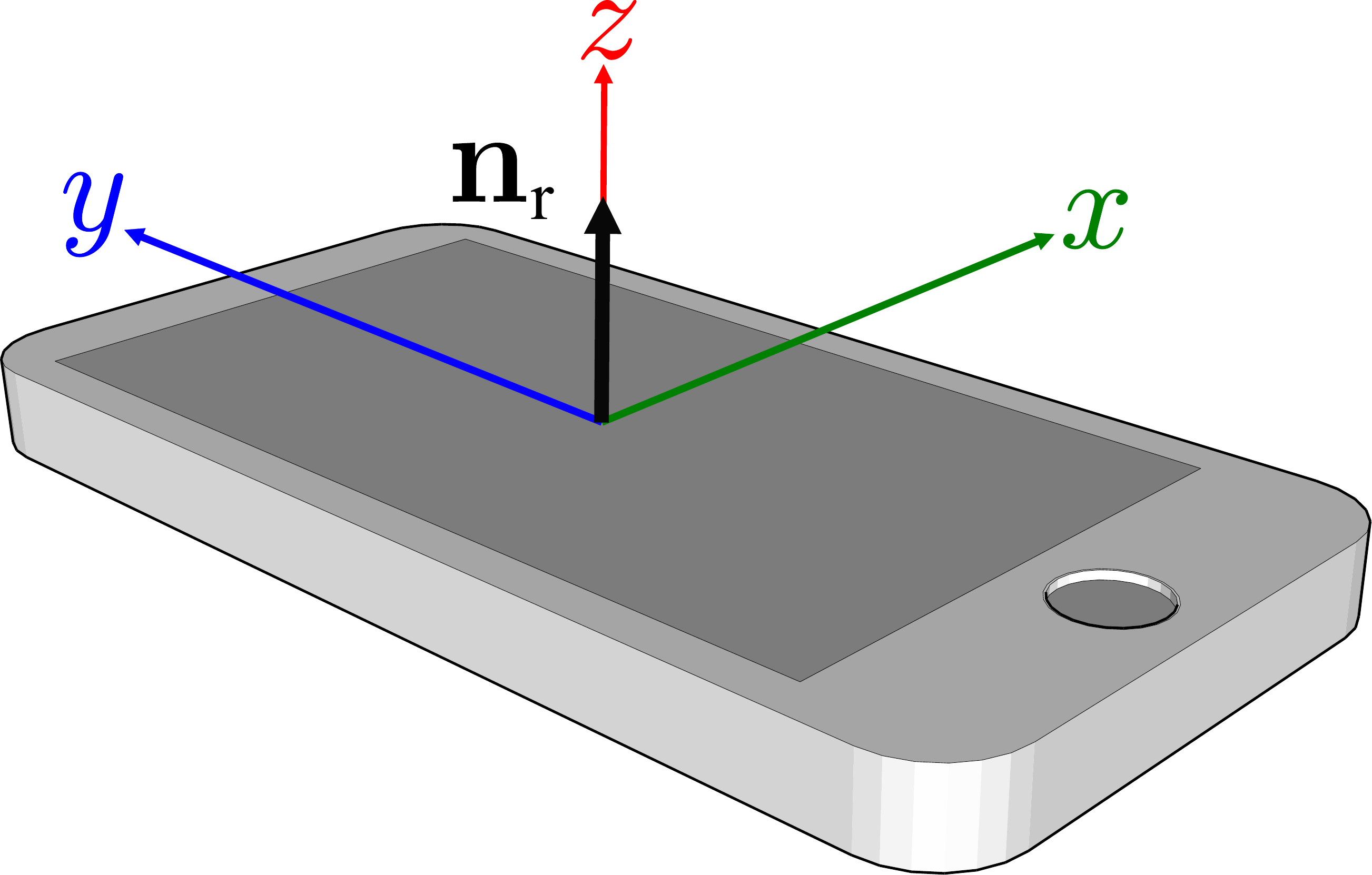}
            \vspace{-0.3cm}
			\caption{}
		\end{subfigure}%
		~
		\begin{subfigure}[b]{0.5\columnwidth}
			\centering
			\includegraphics[width=0.5\columnwidth,draft=false]{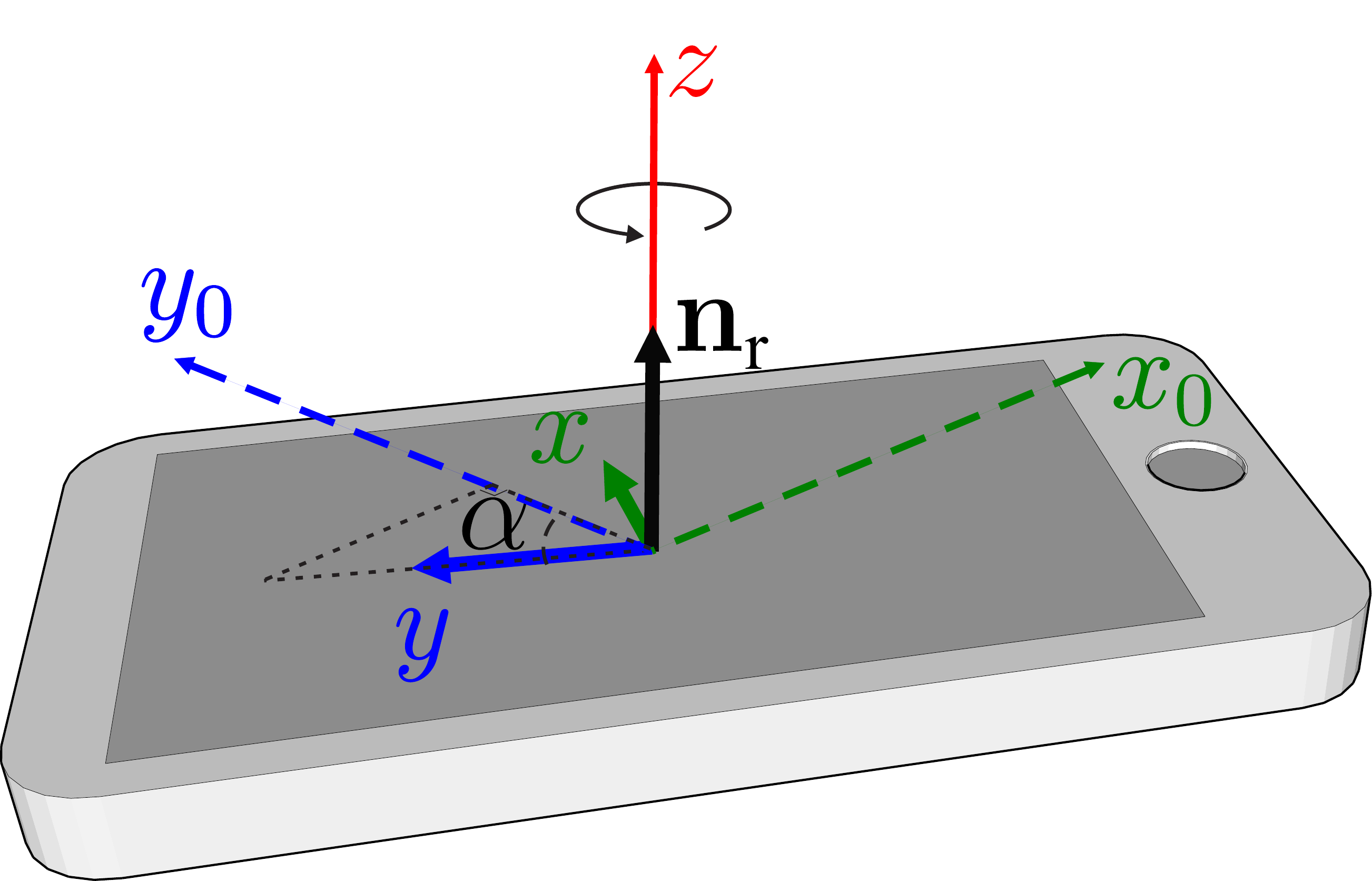}
            \vspace{-0.3cm}
			\caption{}
		\end{subfigure}\\
		\begin{subfigure}[b]{0.5\columnwidth}
			\centering
			\includegraphics[width=0.5\columnwidth,draft=false]{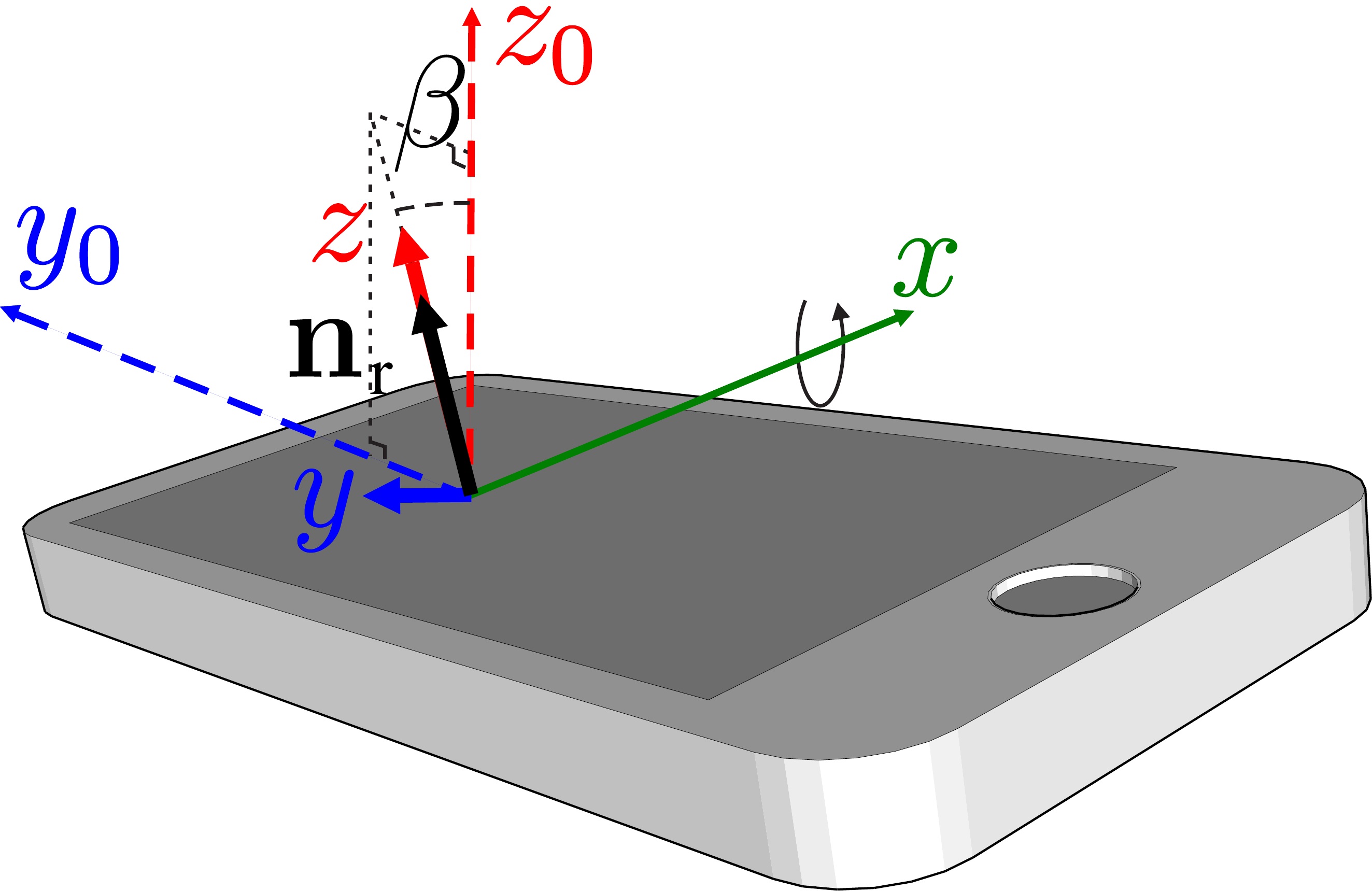}
            \vspace{-0.2cm}
			\caption{}
		\end{subfigure}%
		~
		\begin{subfigure}[b]{0.5\columnwidth}
			\centering
			\includegraphics[width=0.5\columnwidth,draft=false]{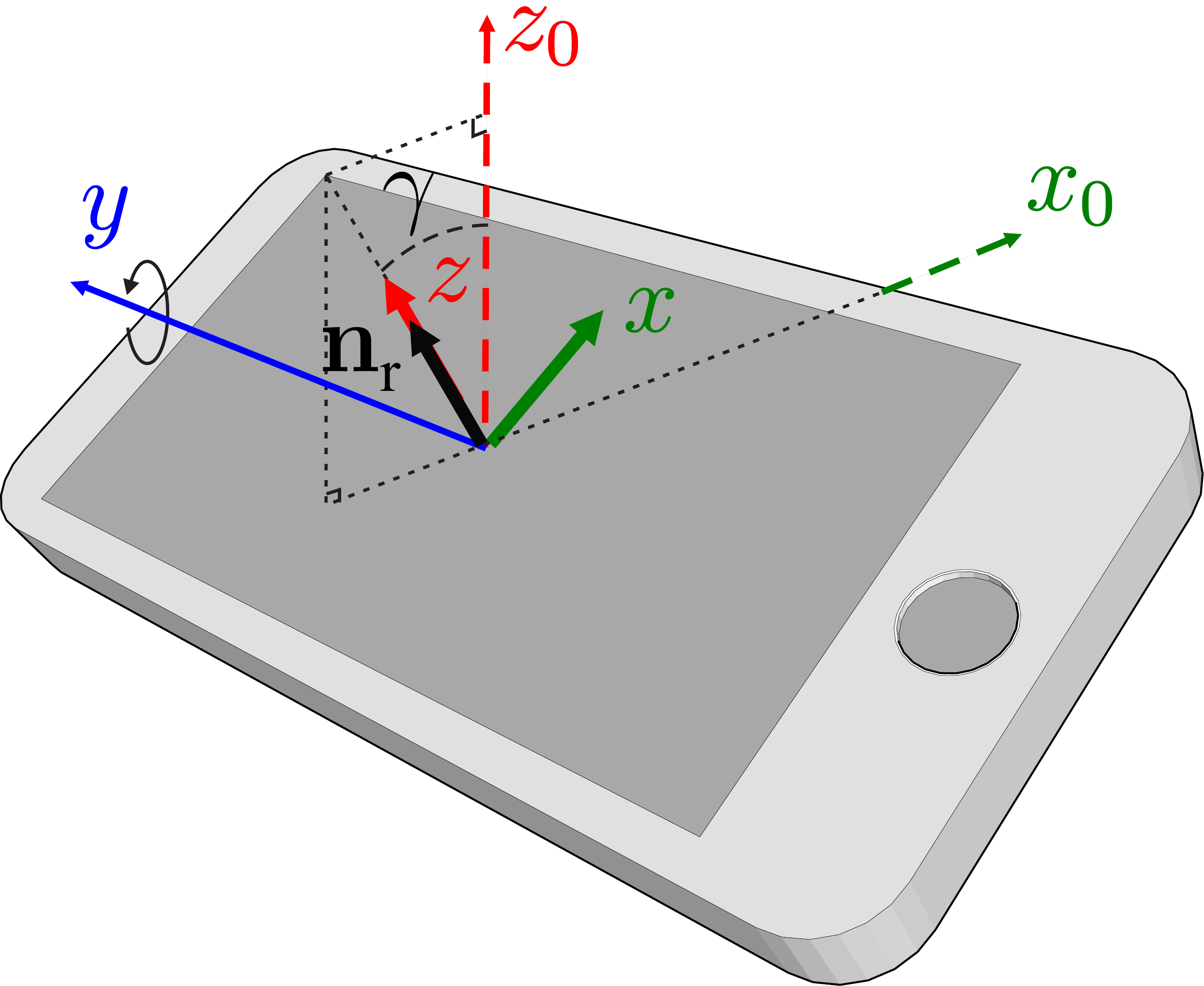}
            \vspace{-0.2cm}
			\caption{}
		\end{subfigure}
        \vspace{-0.7cm}
		\caption{Orientations of a mobile device: (a) normal position, (b) yaw rotation with angle $\alpha$, about the $z$-axis (c) pitch rotation with angle $\beta$, about the $x$-axis and (d) roll rotation with angle $\gamma$, about the $y$-axis.}
		\vspace{-0.5cm}
		\label{figori}
	\end{figure}
    
For collecting the measurements, we have asked $40$ participants to take part in the experiment while they were working with their cellphones. The application ``physics toolbox sensor suite" \cite{androidapp} was used to record the orientation data of yaw, pitch and roll while users were doing normal activities like browsing or watching a video stream. Measurements were recorded for static and mobile users (sitting and walking activities, respectively). More details about the data measurement can be found in \cite{MDSArxiv2018Orientation}. 

The histograms for angles $\alpha$, $\beta$, and $\gamma$, obtained from experimental measurements, along with the Laplace and Gaussian fitted distributions are shown in Figs. \ref{sitting} and \ref{walking}, respectively for sitting and walking activities. It can be seen that the distributions are well fitted with a Laplace distribution for sitting activities while histograms are more close to a Gaussian distribution for walking activities. The mean and variance for each case is noted in Table~\ref{distfit}. We used Kolmogorov-Smirnov distance (KSD) and kurtosis to evaluate the similarity of the collected measurements with the considered distributions \cite{murray1972theory}. The two-sample KSD is the maximum absolute distance between the cumulative distribution functions (CDFs) of two distributions. Small values of KSD (close to zero) confirm more similarity between the distributions. Kurtosis of the random variable, $X$, is defined as ${\rm{Kur[X]}}=\frac{\mathbb{E}[(x-\mu)^4]}{\sigma^4}$, where $\mu$ and $\sigma$ are the mean and variance of the random variable $X$. The kurtosis of Laplace and Gaussian distributions are $6$ and $3$, respectively. 

\begin{figure}[!t]
		\centering
		\begin{subfigure}[b]{0.3\columnwidth}
			\centering
			\includegraphics[width=\columnwidth,draft=false]{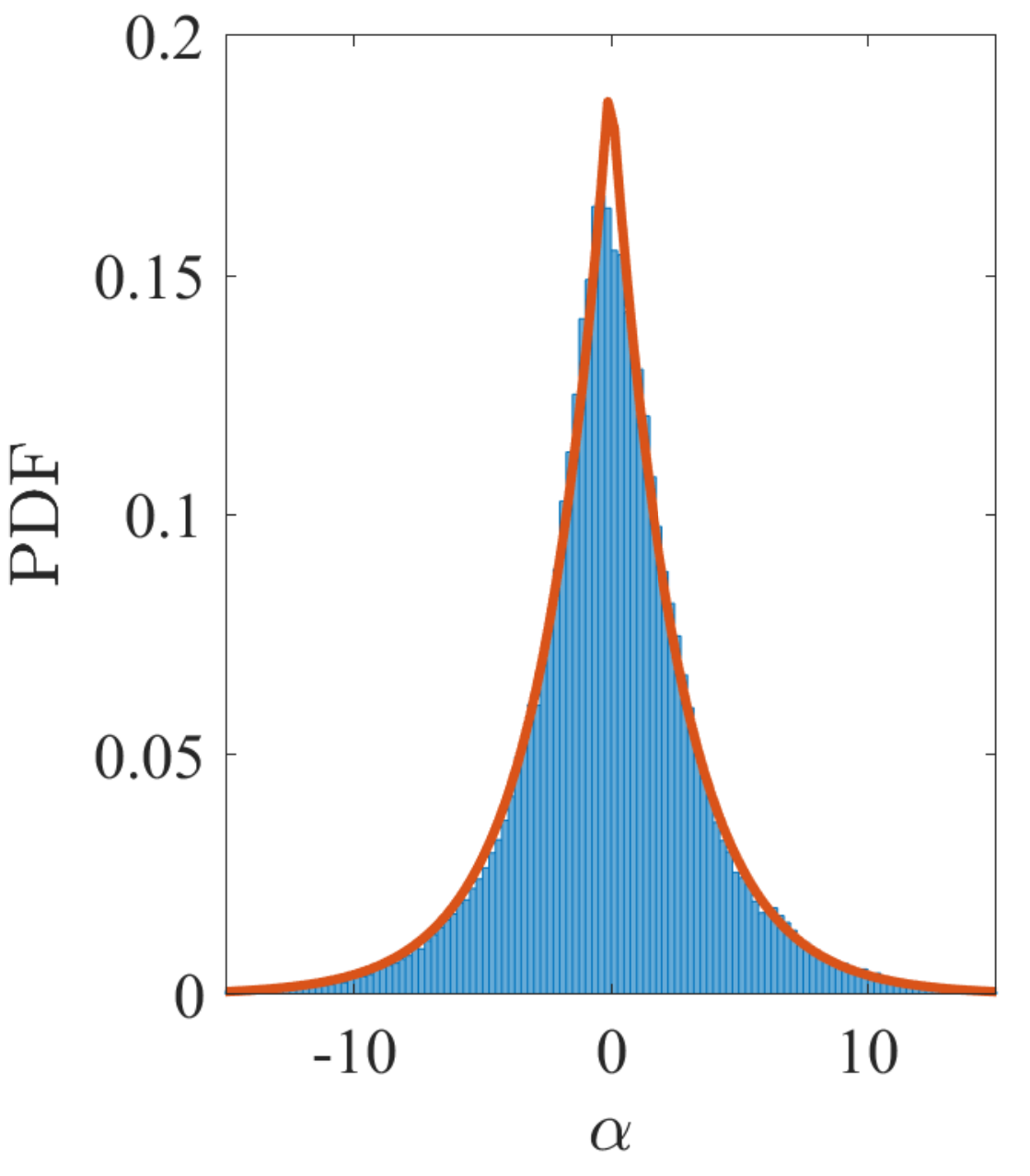}
				\label{alpha_sit}
		\end{subfigure}%
		~
		\begin{subfigure}[b]{0.3\columnwidth}
			\centering
			\includegraphics[width=\columnwidth,draft=false]{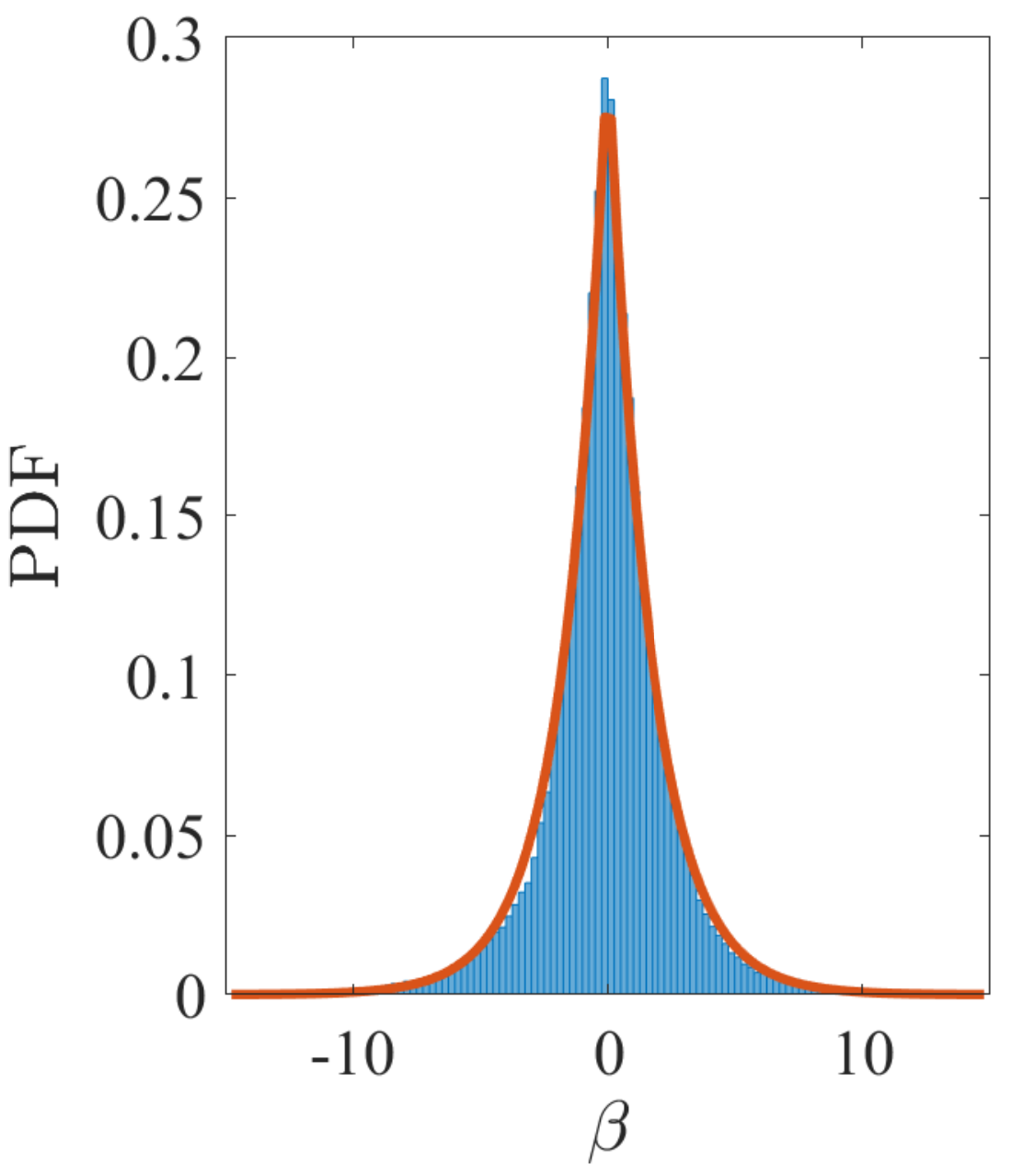}
				\label{beta_sit}
		\end{subfigure}
				~
				\begin{subfigure}[b]{0.3\columnwidth}
					\centering
					\includegraphics[width=\columnwidth,draft=false]{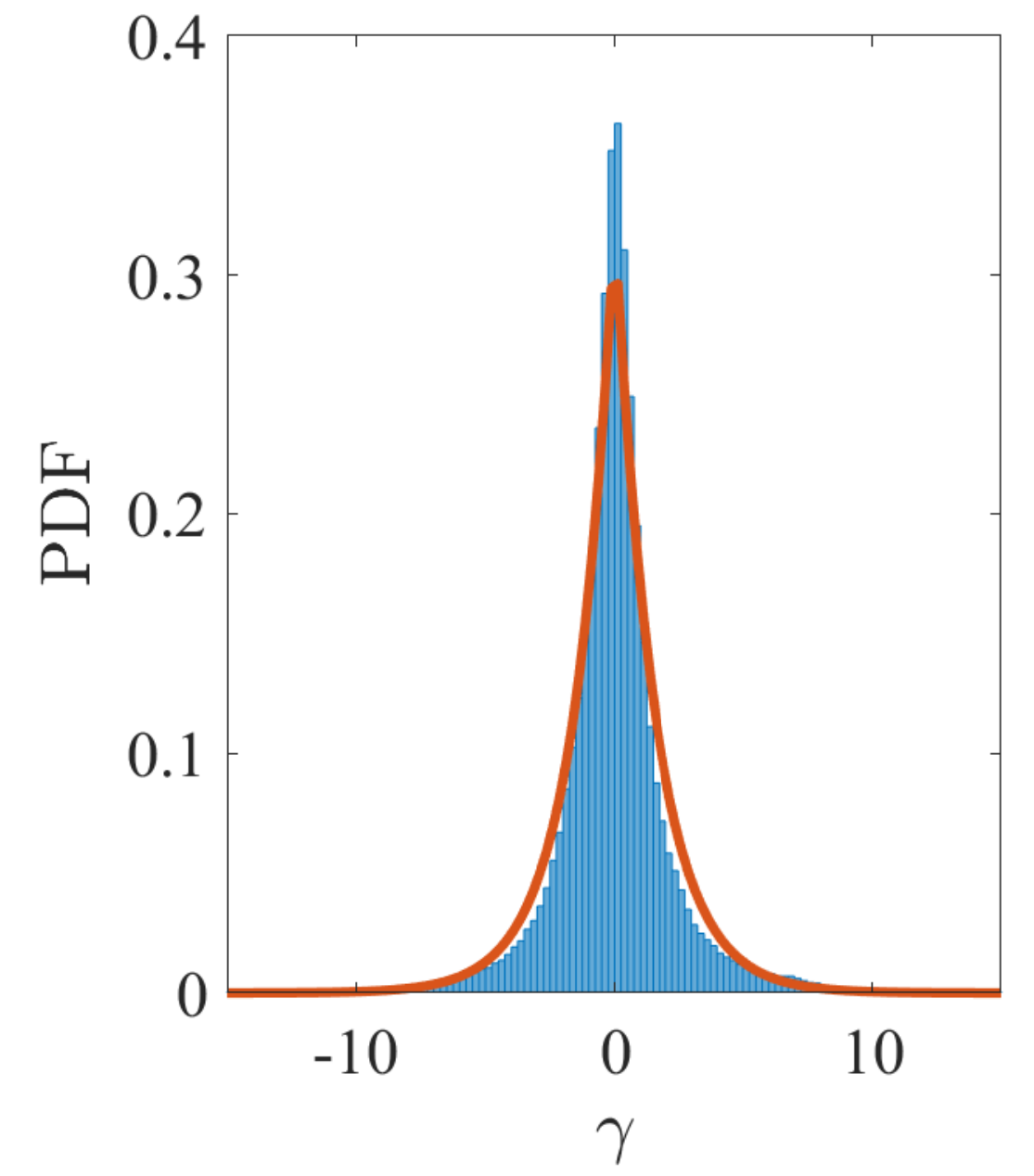}
						\label{gamma_sit}
				\end{subfigure}
                \vspace{-1cm}
				\caption{Histograms of orientation angles $\alpha$, $\beta$, and $\gamma$ for sitting activities obtained from experimental measurements.}
				\vspace{-0.6cm}
				\label{sitting}
\end{figure}
\begin{figure}[!t]
		\centering
		\begin{subfigure}[b]{0.3\columnwidth}
			\centering
			\includegraphics[width=\columnwidth,draft=false]{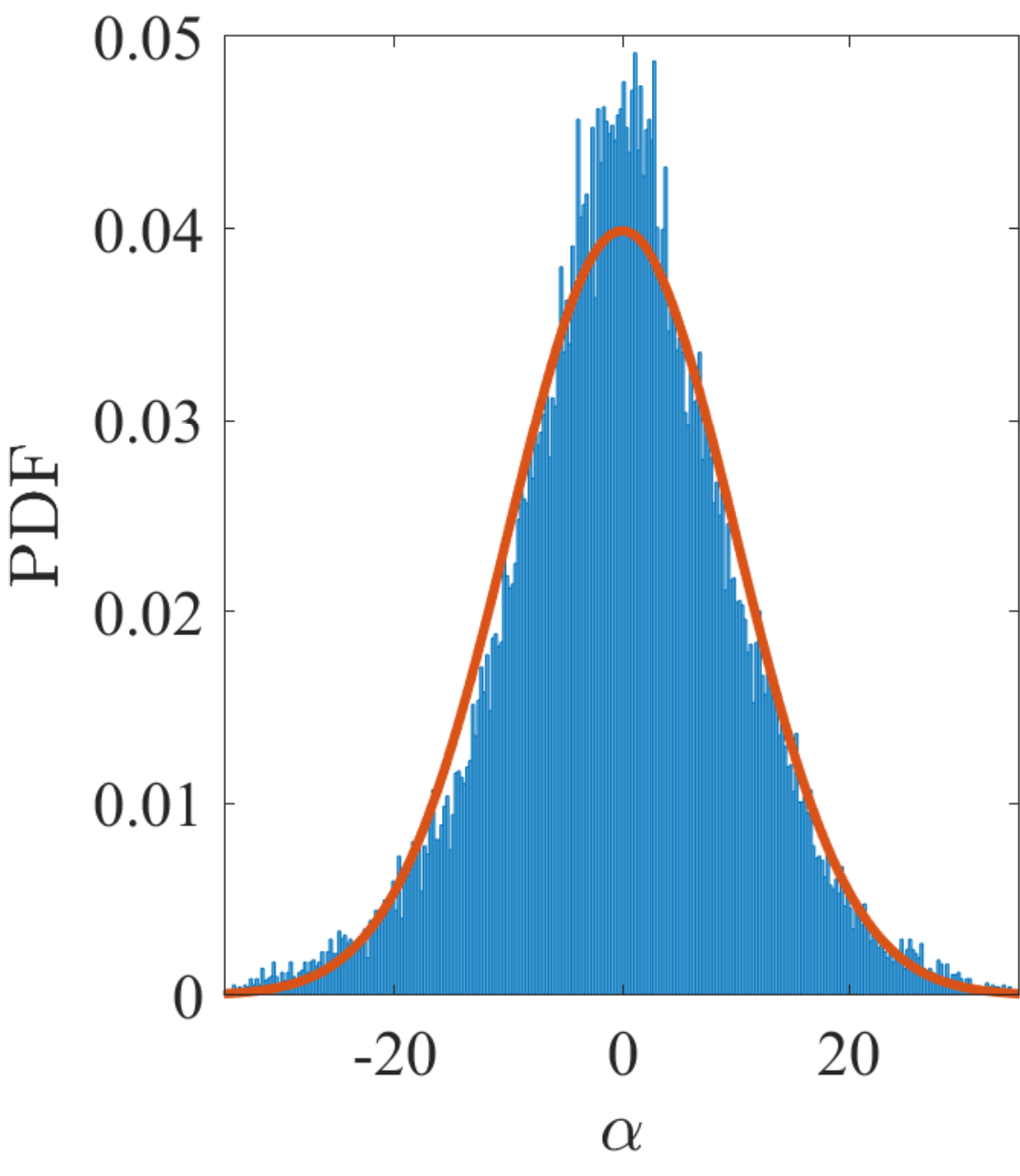}
				\label{alpha_walk}
		\end{subfigure}%
		~
		\begin{subfigure}[b]{0.3\columnwidth}
			\centering
			\includegraphics[width=\columnwidth,draft=false]{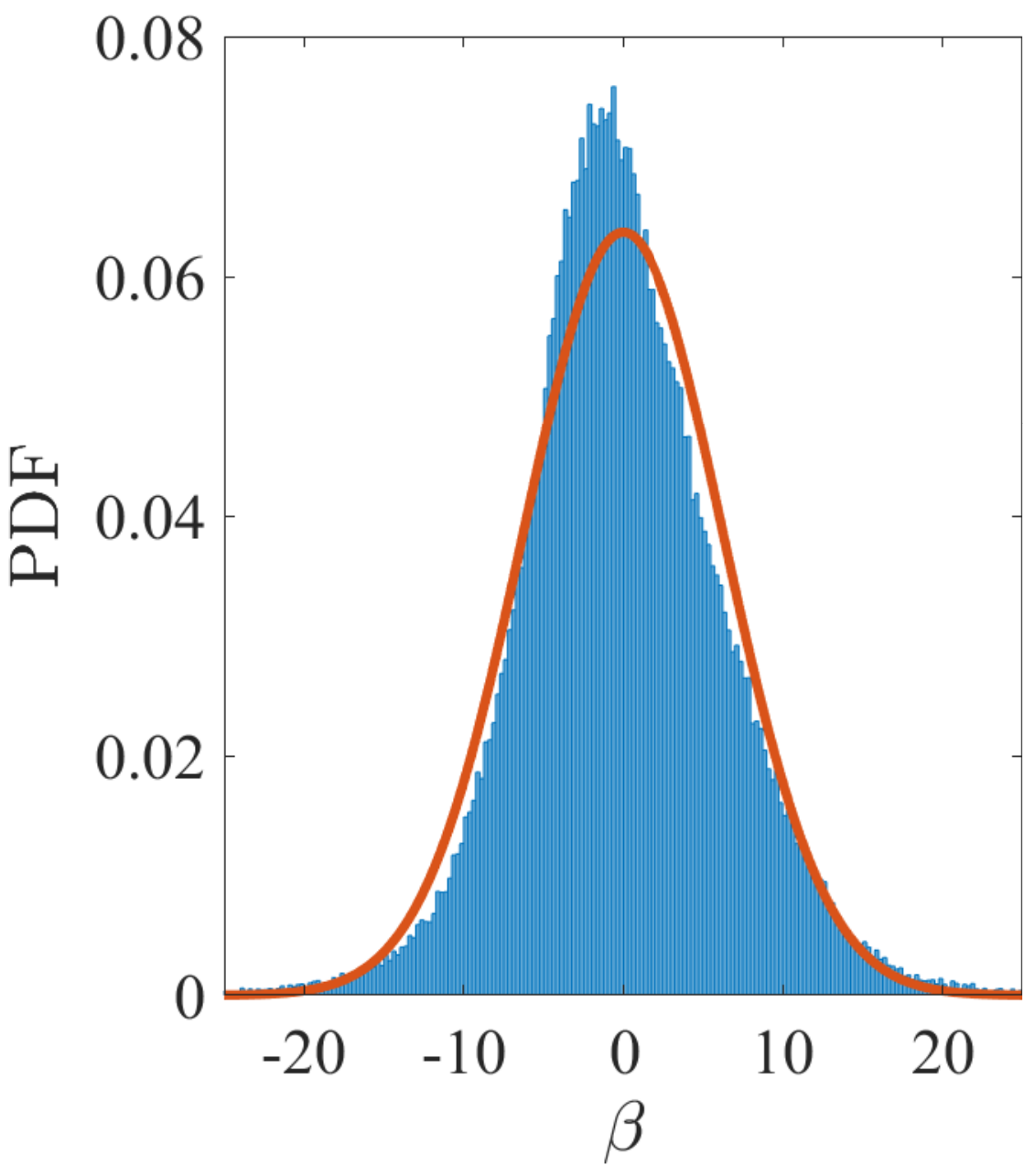}
				\label{beta_walk}
		\end{subfigure}
				~
				\begin{subfigure}[b]{0.3\columnwidth}
					\centering
					\includegraphics[width=\columnwidth,draft=false]{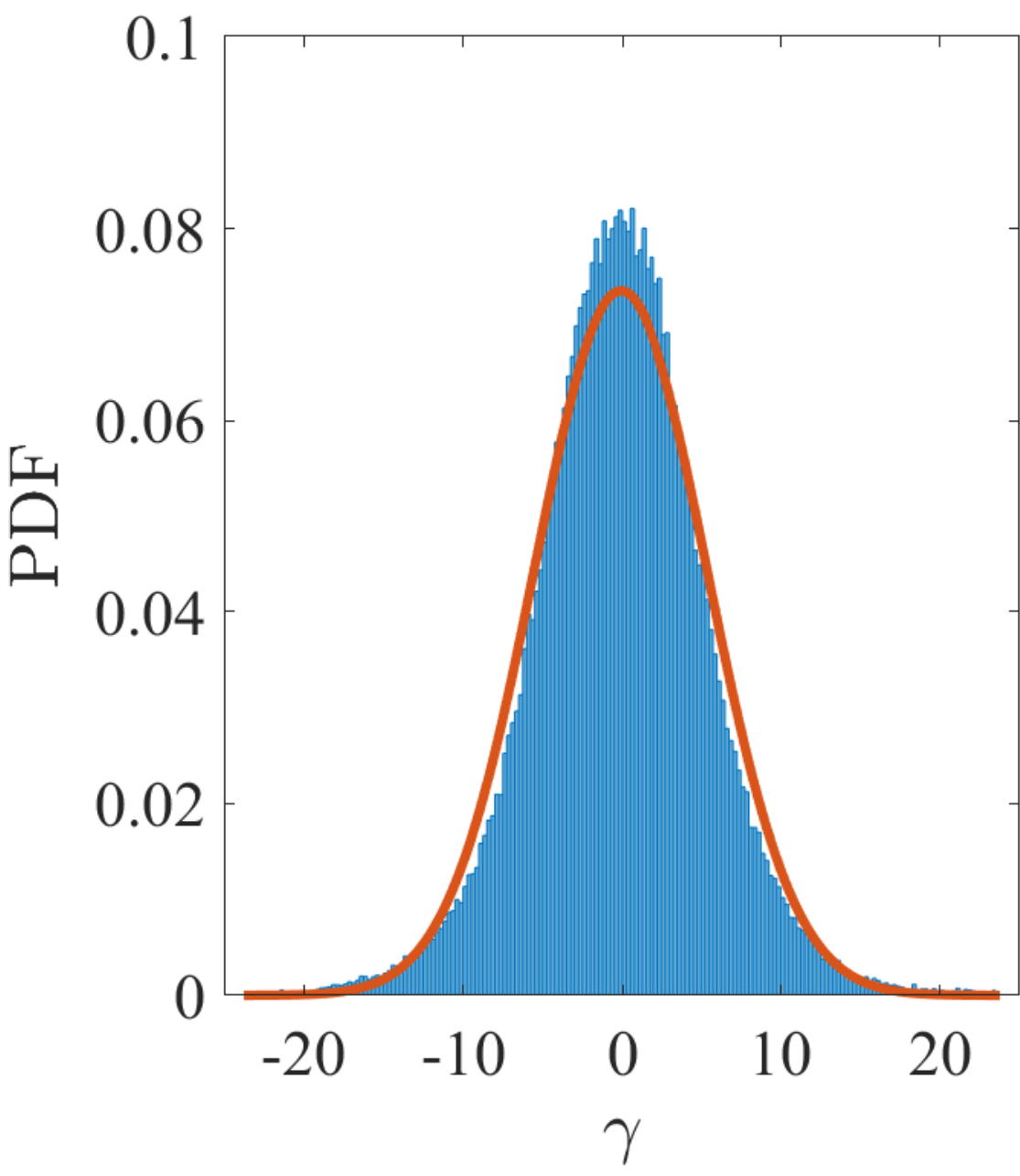}
						\label{gamma_walk}
				\end{subfigure}
                \vspace{-1cm}
				\caption{Histograms of orientation angles $\alpha$, $\beta$, and $\gamma$ for walking activities obtained from experimental measurements.}
			\vspace{-0.6cm}
				\label{walking}
\end{figure}

The statistics given in Table~\ref{distfit} can be used to model the device random orientation. The parameter $\Omega$ shows the user direction that will be explained in more detail in Section IV. Based on this data, the channel matrix at different positions with a random orientation can be found. In other words, random angles are generated using the statistics in Table \ref{distfit} (Gaussian for walking and Laplace for sitting), and then, the channel gain can be obtained using \eqref{h_LOS} and \eqref{h_NLOS}. The same method as described in \cite{ArdimasOFDM} is used to obtain the coherence time of $\alpha$, $\beta$ and $\gamma$. The values are presented in Table~\ref{distfit}. Since this coherence time is much greater than the transmitted symbol time, $T_{\rm{s}}$, and it includes hundreds or thousands of symbols, the channel fading can be considered as block fading. In other words, it is assumed that the channel matrix is estimated and is known for each block of data.

\begin{table}[t!]
	\caption{\small Statistics of orientation measurement.}
    \vspace{-0.3cm}
	\label{distfit}
	\centering
	\begin{tabular}{c c c c c c c}
		\hline
		\hline
		 & \multicolumn{3}{c}{Sitting} & \multicolumn{3}{c}{Walking} \\
		& $\alpha$& $\beta$ & $\gamma$ &  $\alpha$& $\beta$ & $\gamma$ \\
		\hline
		Mean & $\Omega$-90 &40.78& -0.84& $\Omega$-90 & 28.81& -1.35\\
		Standard deviation& 3.67& 2.39 &2.21 & 10&3.26& 5.42\\
		Kurtosis& 6.12& 7.97 &11.82 & 3.47& 3.84& 4.06\\
		Gaussian KSD& 0.07& 0.09 &0.13 & 0.02& 0.03& 0.02\\
		Laplace KSD& 0.01& 0.01 &0.04 & 0.04& 0.06& 0.05\\
        Coherence Time& 0.342& 0.377 &0.331 & 0.131& 0.176& 0.142\\
		\hline
		\hline
	\end{tabular}
    \vspace{-0.3cm}
\end{table}
\vspace{-0.3cm}
\subsection{Link Blockage Modeling}
Due to the nature of OWC, the link between a pair of Tx and Rx can be blocked by an opaque object such as a human body. In this study, we only consider the blockage due to human body or other similar objects which can be modeled as rectangular prisms. It is shown in \cite{BlockageGlob03} that MIMO can help the optical wireless networks to be robust against blockage because transmit or receive diversity is exploited. Here the model for link blockage is introduced which is used throughout the paper.

In this study, we model a human body as a rectangular prism of length, $L_{\rm b}$, width, $W_{\rm b}$, and height, $H_{\rm b}$. Two types of blockers are assumed, non-user blockers and user-blockers. The former is due to the other people or objects in the room while the latter is due to the user who is using the actual UE, also known as self-blockage. Thus, one user-blocker is considered in the direction that the user is facing, and other non-user blockers' directions are chosen from a uniform distribution of $\mathcal{U}[0^{\circ},360^{\circ}]$. Fig.~\ref{blockage} shows the blockage model and the relevant parameters that are considered in this study. 
The density of a non-user blocker is denoted by $\kappa_{\rm b}$, which is the number of non-user blockers per room area. It is assumed that non-user blockers are uniformly distributed in the room. The direction and location of the self-blocker are obtained based on the direction and location of the UE. It is assumed that the users keep the UE at a distance of $d_{\rm p}$ away from themselves. 





\begin{figure}
		\begin{center}
			\includegraphics[width=0.75\columnwidth,draft=false]{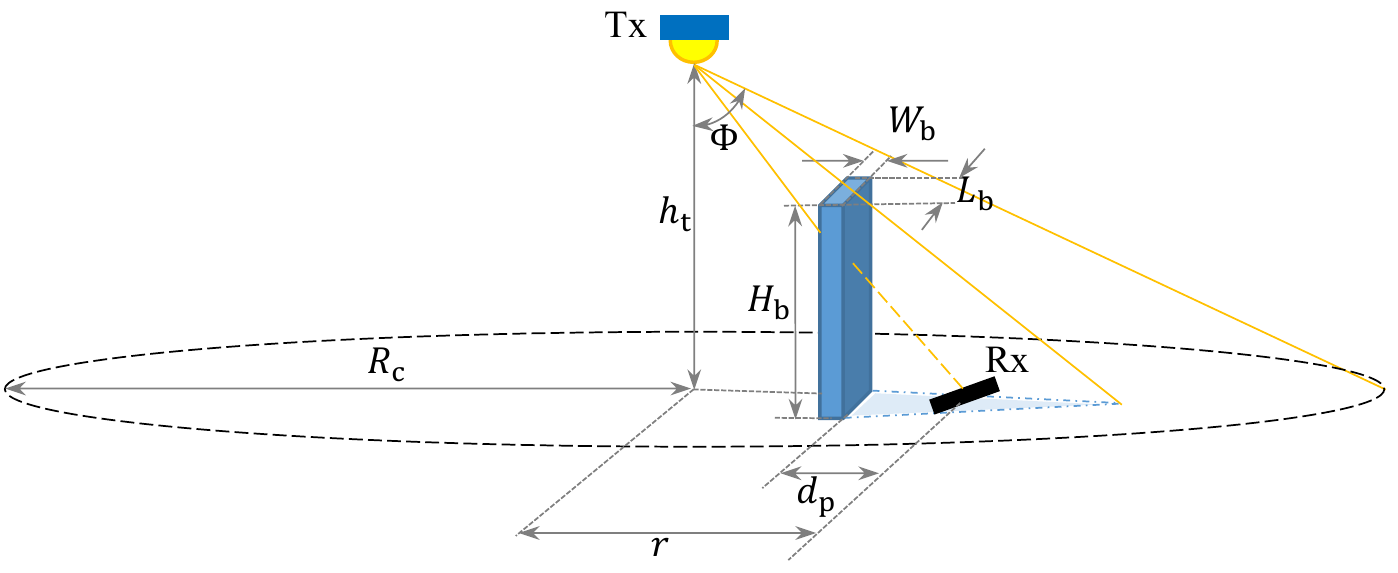}
			\caption{Geometry of link blockage.}
			\label{blockage}
		\end{center}
        \vspace{-0.7cm}
\end{figure}

The availability of the channel state information (CSI) is essential for the implementation of SM, as explained in Section \ref{sec_sm}. The measurement results in Table \ref{distfit} show that the coherence time for the orientation angles are on the order of several hundred milliseconds, and we assume that the coherence time for the channel matrix is also similar. This was confirmed in \cite{MDSArxiv2018Orientation}. Therefore, it can be assumed that the channel gains are known for each transmitted data block of length smaller than the channel coherence time. However, there can be errors in the estimation of the channel gain $h_{i,j}$, which could be independent of the channel gain itself \cite{basar2012performance,rajashekar2014reduced,gifford2005diversity}. An estimation error causes additional errors depending on the estimation method used. However, we assume that the estimation error can be ignored throughout this paper. An individual study can be carried out in the future to investigate the effect of the channel estimation error in different conditions and for various estimation methods.
\section{Performance of Downlink}\label{sec_downlink}
\subsection{System Configuration}
In this paper, we consider a typical indoor environment. Although the results may change slightly in other scenarios, it is expected that the same behavior will be observed provided that the main characteristics of the environment, such as transmitter separation, the room dimensions and the ceiling height, etc., do not vary dramatically. Fig.~\ref{RoomGeo} shows the geometric configuration of the transmitters which are arranged on the vertexes of a square lattice over the ceiling of an indoor environment. We use this configuration throughout the paper, where $16$ LED transmitters, called access points (AP), are considered in a room of $5\times5\times3$ m$^3$. Transmitters are oriented vertically downward, while the receivers may have a random orientation. 
\begin{figure}[!ht]
\centering
\begin{subfigure}[b]{0.5\columnwidth}
\centering
\includegraphics[width=1\columnwidth,draft=false]{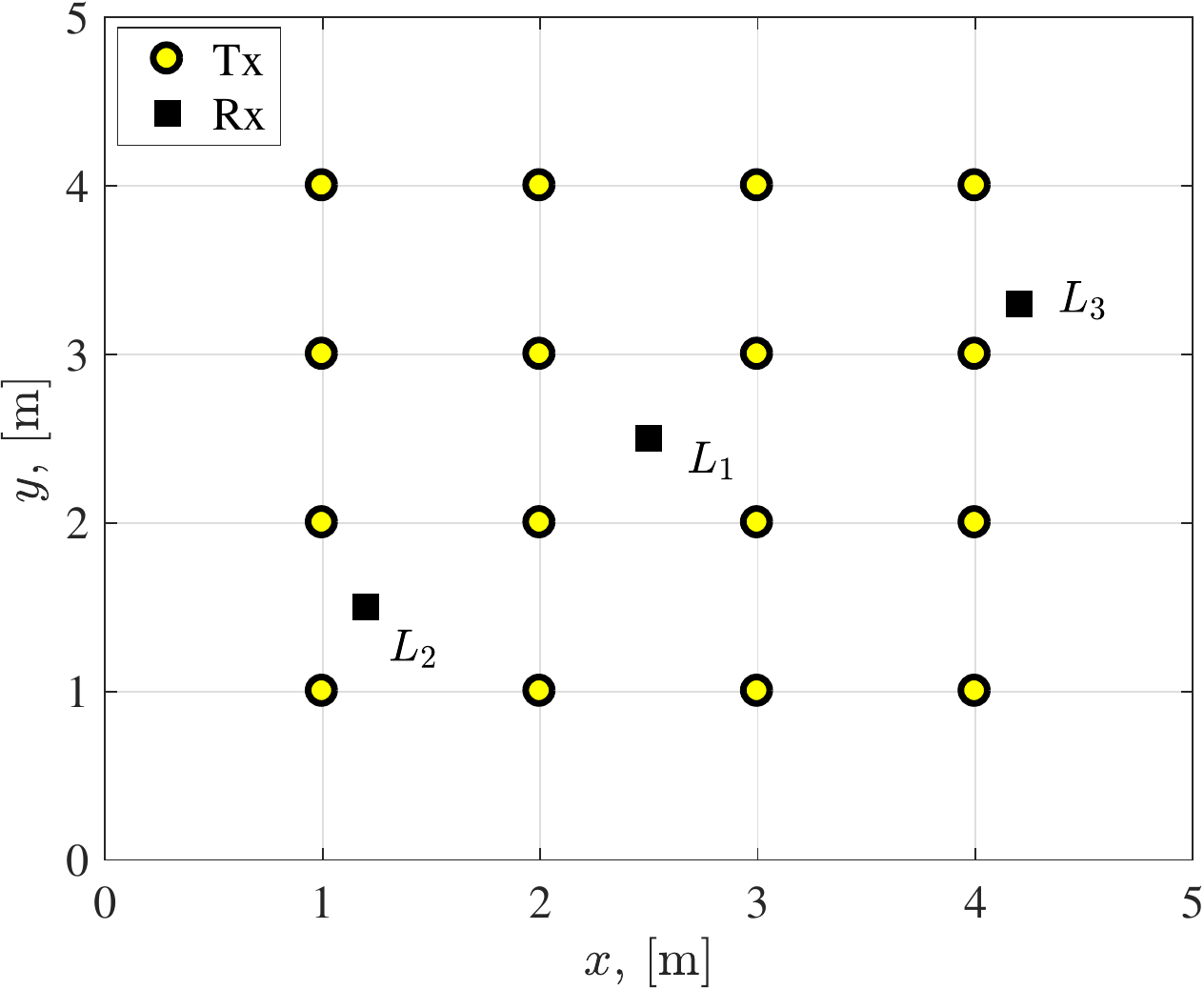}
\caption{Room geometry and transmitters arrangement}
\label{RoomGeo}
\end{subfigure}~
\begin{subfigure}[b]{0.5\columnwidth}
\centering
\includegraphics[width=0.64\columnwidth,draft=false]{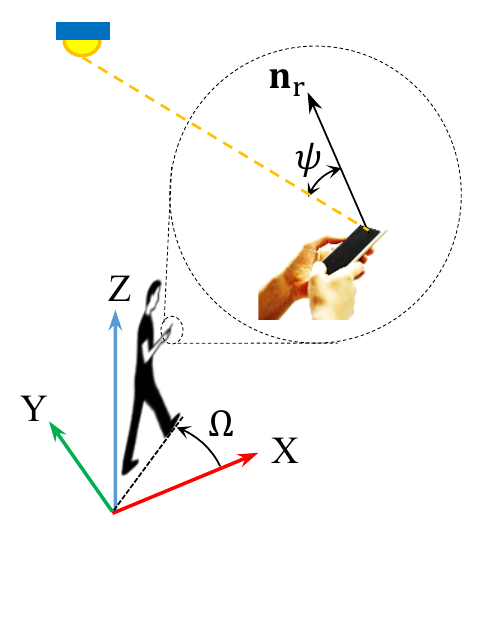}
\caption{User direction}
\label{Direction}
\end{subfigure}\vspace{-0.25cm}
\caption{Illustration of room geometry and user direction.}
\label{room}
\vspace{-0.25cm}
\end{figure}

Let us define $\Omega$ as the facing or movement direction of a user while sitting or walking, which is measured from the East direction in the Earth coordinate system as shown in Fig.~\ref{Direction}. Since $\alpha$ is measured from the North direction, we have $\mathbb{E}[\alpha]=90-\Omega$, as stated in Table \ref{distfit}. For the rest of the paper, we use $\Omega$ as it provides a better physical view, i.e., facing direction. 


\begin{table}[t]
	\centering
    \vspace{-0.3cm}
	\caption{Simulation Parameters}
    \vspace{-0.3cm}
	\label{TableSimulationParam}
	{\raggedright
		\begin{tabular}{p{125pt}|p{35pt}|p{35pt}|p{125pt}|p{35pt}|p{35pt}}
			\hline
			\parbox{130pt}{\centering{\small Parameter}} & \parbox{40pt}{\centering{\small Symbol}} & \parbox{40pt}{\centering{\small Value}} &\parbox{130pt}{\centering{\small Parameter}} & \parbox{40pt}{\centering{\small Symbol}} & \parbox{40pt}{\centering{\small Value}}\\
			\hline
			\hline
			\parbox{130pt}{\raggedright{\small Receiver FOV}} & \parbox{40pt}{\centering{\small $\Psi$}} & \parbox{40pt}{\centering{\small $60^\circ$}} &\parbox{130pt}{\raggedright{\small Ceiling height}} & \parbox{40pt}{\centering{\small $h_\mathrm{z}$}} & \parbox{40pt}{\centering{\small 3 m}}\\
			\hline
			\parbox{130pt}{\raggedright{\small LED half-power semiangle}} & \parbox{40pt}{\centering{\small $\Phi$}} & \parbox{40pt}{\centering{\small $60^\circ$}} &\parbox{130pt}{\raggedright{\small Reflectivity factor of walls}} & \parbox{40pt}{\centering{\small $\rho_\mathrm{w}$}} & \parbox{40pt}{\centering{\small $0.6$}}\\
			\hline
			\parbox{130pt}{\raggedright{\small PD responsivity}} & \parbox{40pt}{\centering{\small $R_{\rm{PD}}$}} & \parbox{40pt}{\centering{\small $1$ A/W }} &\parbox{130pt}{\raggedright{\small Reflectivity factor of the floor}} & \parbox{40pt}{\centering{\small $\rho_\mathrm{f}$}} & \parbox{40pt}{\centering{\small $0.2$}}\\
			\hline
			\parbox{130pt}{\raggedright{\small Physical area of a PD}} & \parbox{40pt}{\centering{\small $A$}} & \parbox{40pt}{\centering{\small $0.25$ cm$^2$}} &\parbox{130pt}{\raggedright{\small Reflectivity factor of the ceiling}} & \parbox{40pt}{\centering{\small $\rho_\mathrm{c}$}} & \parbox{40pt}{\centering{\small $0.8$}} \\
			\hline
			\parbox{130pt}{\raggedright{\small UE height (sitting)}} & \parbox{40pt}{\centering{\small $h_\mathrm{r}$}} & \parbox{40pt}{\centering{\small $0.8$ m}} &\parbox{130pt}{\raggedright{\small Length of the blockers}} & \parbox{40pt}{\centering{\small $L_{\rm b}$}} & \parbox{40pt}{\centering{\small $0.7$ m}} \\
			\hline
            			\parbox{130pt}{\raggedright{\small UE height (walking)}} & \parbox{40pt}{\centering{\small $h_\mathrm{r}$}} & \parbox{40pt}{\centering{\small $1.4$ m}} &\parbox{130pt}{\raggedright{\small Width of the blockers}} & \parbox{40pt}{\centering{\small $W_{\rm b}$}} & \parbox{40pt}{\centering{\small $0.2$ m}} \\
			\hline
            \parbox{130pt}{\raggedright{\small Light source height}} & \parbox{40pt}{\centering{\small $h_\mathrm{t}$}} & \parbox{40pt}{\centering{\small $2.95$ m}} & \parbox{130pt}{\raggedright{\small Height of the blockers}} & \parbox{40pt}{\centering{\small $H_{\rm b}$}} & \parbox{40pt}{\centering{\small $1.75$ m}}\\
			\hline
            \parbox{130pt}{\raggedright{\small Room dimensions}} & \parbox{40pt}{\centering{\small $h_\mathrm{x}\times h_\mathrm{y}$}} & \parbox{50pt}{{\small $5$ m$\times 5$ m}} &\parbox{130pt}{\raggedright{\small UE distance from user}} & \parbox{40pt}{\centering{\small $d_{\rm p}$}} & \parbox{40pt}{\centering{\small $0.3$ m}} \\
            \hline
			\hline
		\end{tabular}
	}
	\vspace{-0.5cm}
\end{table}

Parameters used throughout the paper are shown in Table \ref{TableSimulationParam}. The dimensions of the smartphone are $14\times 7\times 1$ cm$^3$. As shown in Fig. \ref{phones}, the PDs on the screen are placed 1 cm from the top edge, one PD for MDR at the center is considered, and 4 PDs for SR are uniformly distributed. For MDR, the PD associated with $\mathbf{n}_3$ is placed at the center of the corresponding side, and PDs shown by $\mathbf{n}_2$ and $\mathbf{n}_4$ are placed 1.5 cm from the top edge. 

\vspace{-0.1cm}

\subsection{The Effect of Blockage and Random Orientation}
Here, we investigate the effect of low and high density receiver blockage, fixed orientation versus random orientation, and the effect of the NLOS channel component. First, an example scenario is considered at the middle of the room (i.e., location $L_1$ in Fig. \ref{RoomGeo}) with fixed AP allocation with $N_\mathrm{a}=4$ and spectral efficiency $R=5$ bits/sec/Hz. The user direction is $\Omega=90^{\circ}$. The APs are determined by measuring the received power from each AP, and the strongest $N_\mathrm{a}=4$ APs, i.e., the ones corresponding to the highest received power at the user position, are selected. The results are shown in Fig. \ref{BERSRMDR} for both BER approximation in \eqref{bersm} (solid lines) and Monte-Carlo simulations (markers). Note that the statistics of random orientation for sitting activities are used according to Table I for Laplace distribution.

\begin{figure}
		\begin{subfigure}[b]{0.5\columnwidth}
			\includegraphics[width=1\columnwidth,draft=false]{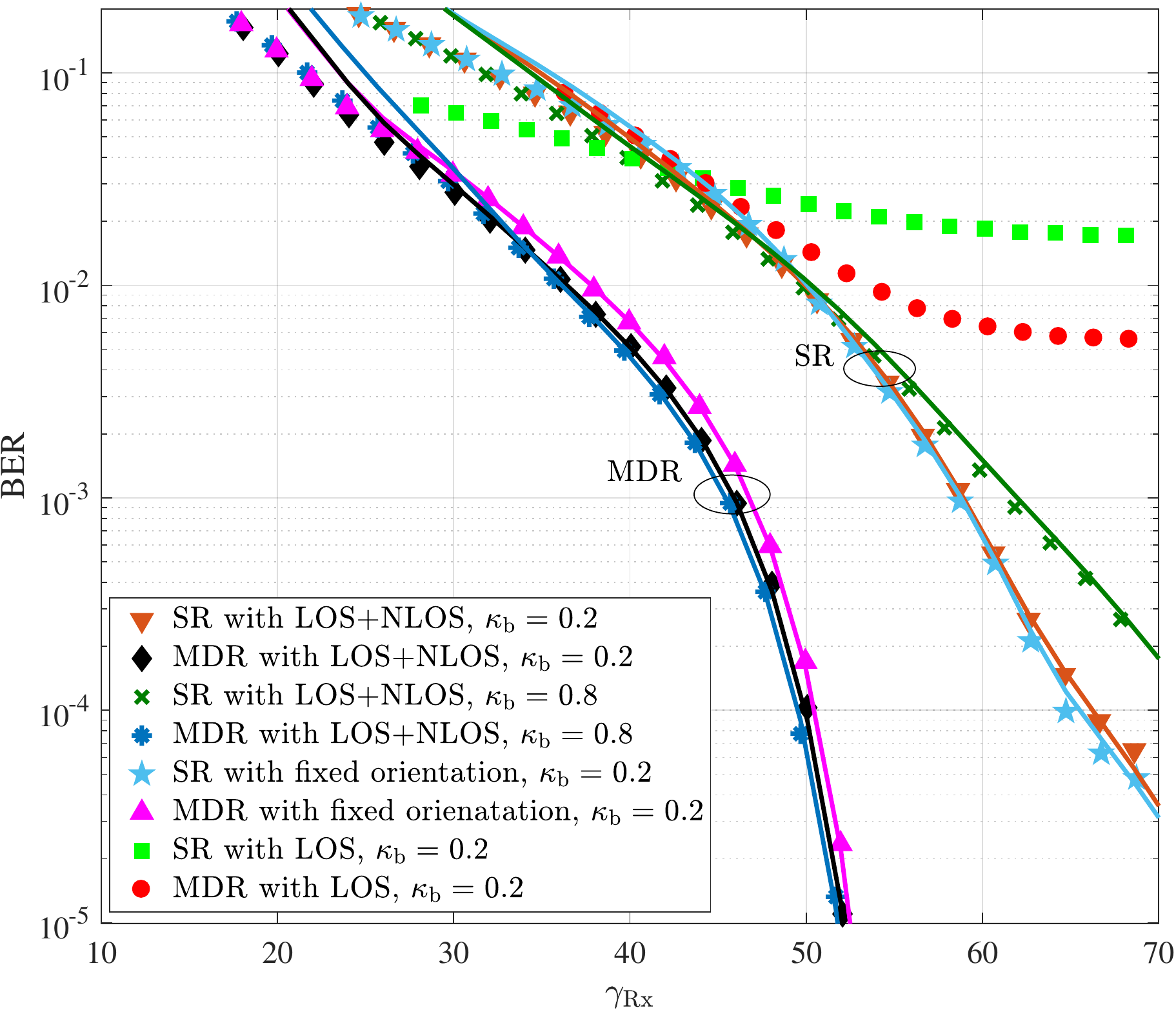}\vspace{-0.2cm}
			\caption{}
			\label{BERSRMDR}
		\end{subfigure}~
        	\begin{subfigure}[b]{0.5\columnwidth}
			\includegraphics[width=1\columnwidth,draft=false]{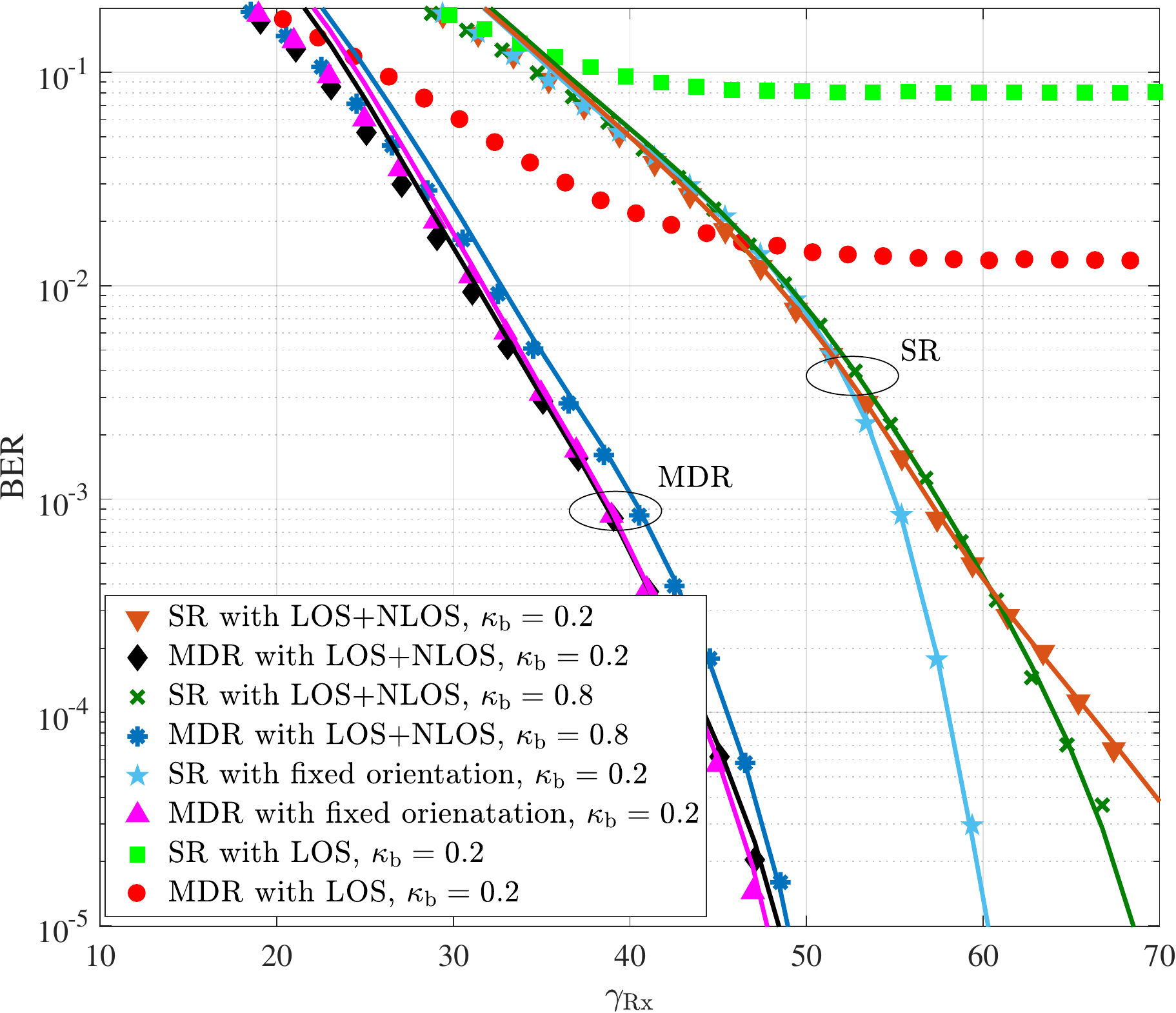}\vspace{-0.2cm}
			\caption{}
            \label{BERSRMDR2}
		\end{subfigure}\vspace{-0.3cm}
        \caption{Performance comparison of SR and MDR for UE's location of a) $L_1$ and direction of $\Omega=90^{\circ}$ b) $L_3$ and direction of $\Omega=180^{\circ}$. Marks denote Monte-Carlo simulation results and solid lines are based on the BER upper bound given in \eqref{bersm}.}
        \vspace{-0.5cm}
\end{figure}

It is observed in Fig. \ref{BERSRMDR} that the simulation results match the BER approximation at around $\mathrm{BER}\geq 10^{-2}$ for both cases. The performance of both SR and MDR is significantly degraded when the NLOS channel gains are ignored because it is highly likely that the channel gains for some of the APs are zero due to blockage or limited FOV. In such a case, the information will be lost, and the BER saturates at high SNR. The BER approximation \eqref{bersm} is not valid in this case and is not shown in Fig. \ref{BERSRMDR}. When the NLOS gain is included, the performance of MDR is always better than SR with SNR gains up to 12 dB at target BER $3.8\times10^{-3}$. For SR, the performance with random orientation is slightly worse than the fixed orientation with blockage parameter $\kappa_\mathrm{b}=0.2$. This happens because the strongest links are chosen at each random orientation, and in the middle of the room almost always some ``good'' links are found which lead to this performance. However, when the blockage parameter is increased, the BER also degrades for SR. On the other hand, both random orientation and increased blockage are beneficial for MDR since the induced randomness increases the differentiability between spatial symbols, which consequently leads to better performance. Overall, it can be seen that the proposed MDR structure is robust against random orientation and blockage, and exhibits superior performance compared to the conventional SR.

Another location, $L_3$ in Fig. \ref{RoomGeo}, is also considered with user direction $\Omega=180^{\circ}$. In this case the user is facing the room and the UE screen is facing the wall. The results are shown in Fig. \ref{BERSRMDR2}. In this scenario the effect of the NLOS channel gain is much more significant. Note that, unlike $L_1$, the channel matrix is non-symmetric at $L_3$, which generally leads to better performance compared to $L_1$. For SR, both random orientation and higher blockage density cause worse performance, and this is severe because ``good'' links are limited in the vicinity of the walls. Again, MDR outperforms SR in any condition, and demonstrates a robust performance against random orientation and blockage. In Fig. \ref{BERSRMDR2}, the random orientation has negligible effect on the performance of MDR, but the link blockage adversely affects the BER performance. This happens because the channel matrix is already non-symmetric, and blockage slightly worsens the channel matrix. The results shown in Figs. \ref{BERSRMDR} and \ref{BERSRMDR2} confirm that the proposed MDR outperforms SR and is robust against random orientation and blockage. However, the exact performance depends on the user location. This will be investigated later in this section, but first the effect of AP selection is studied in the next subsection. 

\subsection{AP Selection}
In the previous subsection, the number of selected APs was fixed and the APs were selected based on the received power at the UE location, user direction and UE orientation. It is expected that the choice of parameter  $N_\mathrm{a}$ can affect the performance of the system with fixed target BER and spectral efficiency. Therefore, the BER performance of MDR and SR are shown in Fig. \ref{BERSRMDR_AP} for position $L_2$ (see Fig. \ref{RoomGeo}) with $\Omega=0^\circ$ for $R=5$ bits/sec/Hz for $N_\mathrm{a}=1,4,16$. It is observed that the BER varies for each choice of $N_\mathrm{a}$. For MDR, $N_\mathrm{a}=16$ achieves the best performance while $N_\mathrm{a}=1$ is the best choice for SR in this specific scenario.

Fig. \ref{BERSRMDR_AP} is an example which highlights the importance of AP selection. A simple method can be defined based on this observation, and an adaptive SM (ASM) is defined. The parameter $N_\mathrm{a}$ is determined at each user position, direction, and UE orientation in order to select the one associated with the minimum energy requirement at the target BER $3.8\times10^{-3}$. This simple method is performed by calculating the required $\gamma_\mathrm{RX}$ based on BER approximation \eqref{bersm} which is a tight approximation for the target BER. 

In order to evaluate the performance of the proposed adaptive SM, the room area is divided into uniformly distributed points that are 25 cm apart in $x$ and $y$ directions. At each point, 24 user directions (every 15$^\circ$) are used, and 500 random orientation angles are generated for each user position and direction. The CDF of the required received SNR over the room is demonstrated in Fig. \ref{CDFSRMDR_AP} for MDR and SR structures. Fixed AP numbers of $N_\mathrm{a}=1$ and 16 are depicted along with the adaptive method, in which the optimum number is determined from $N_\mathrm{a}\in \{1,2,4,8,16\}$ for each user position, direction, and UE orientation. As expected, the minimum energy consumption is achieved by using this simple adaptive method, and MDR significantly outperforms SR. However, it can be seen that the CDF of the received SNR for ASM with MDR method is close to fixed $N_\mathrm{a}=16$. This indicates that whenever the complexity is a limiting factor, fixing $N_\mathrm{a}=16$ can be used for MDR. However, a similar statement is not applicable to SR.

\begin{figure}
\begin{minipage}[t]{.5\textwidth}
  \centering
  \includegraphics[width=1\linewidth]{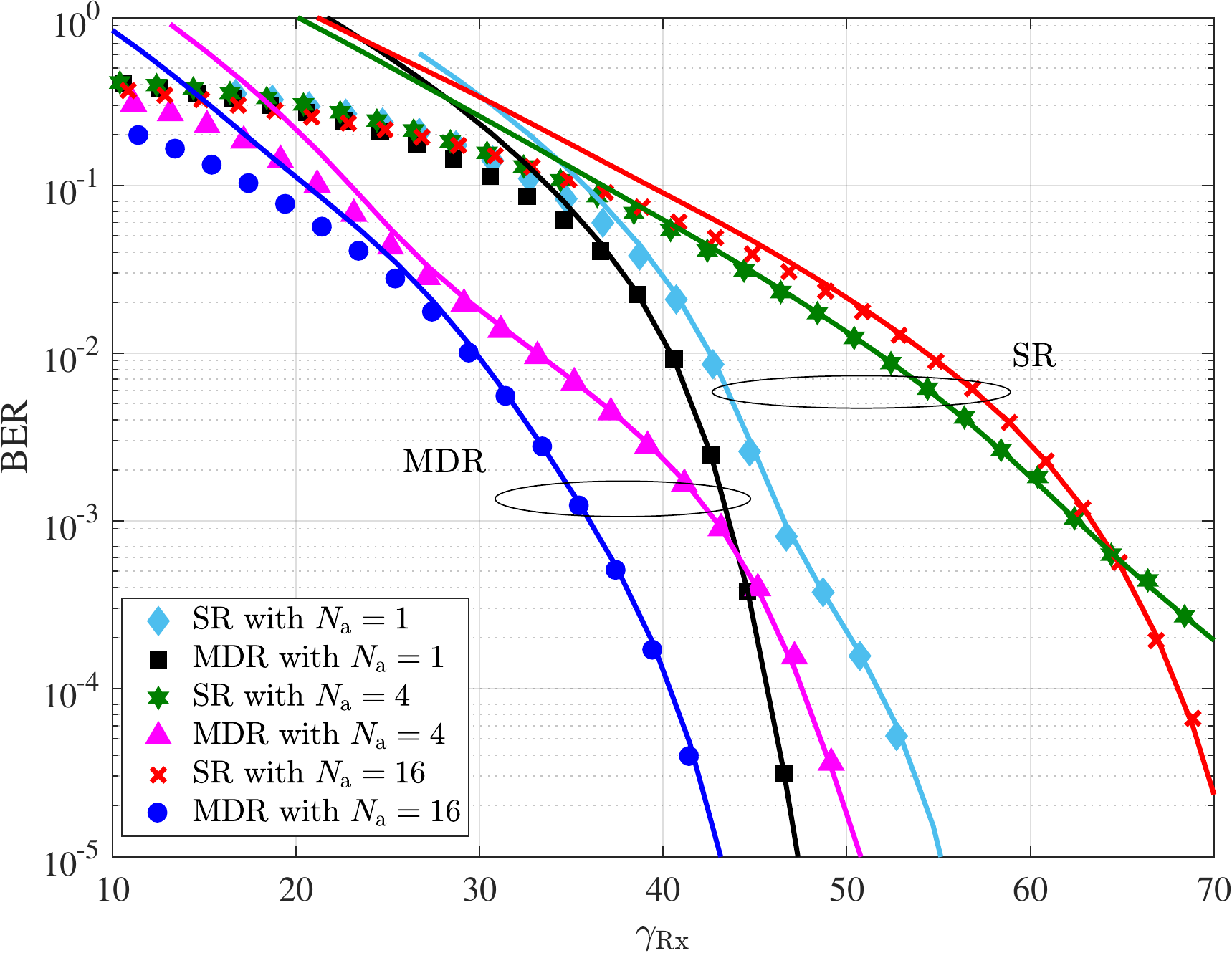}\vspace{-0.2cm}
  \captionof{figure}{Performance comparison of SR and MDR for UE's location of $L_2$ and direction of $\Omega=0^{\circ}$. Markers denote Monte-Carlo simulation and solid lines are based on the BER approximation given in \eqref{bersm}.}
  \label{BERSRMDR_AP}
\end{minipage}~
\begin{minipage}[t]{.5\textwidth}
  \centering
  \includegraphics[width=1\linewidth]{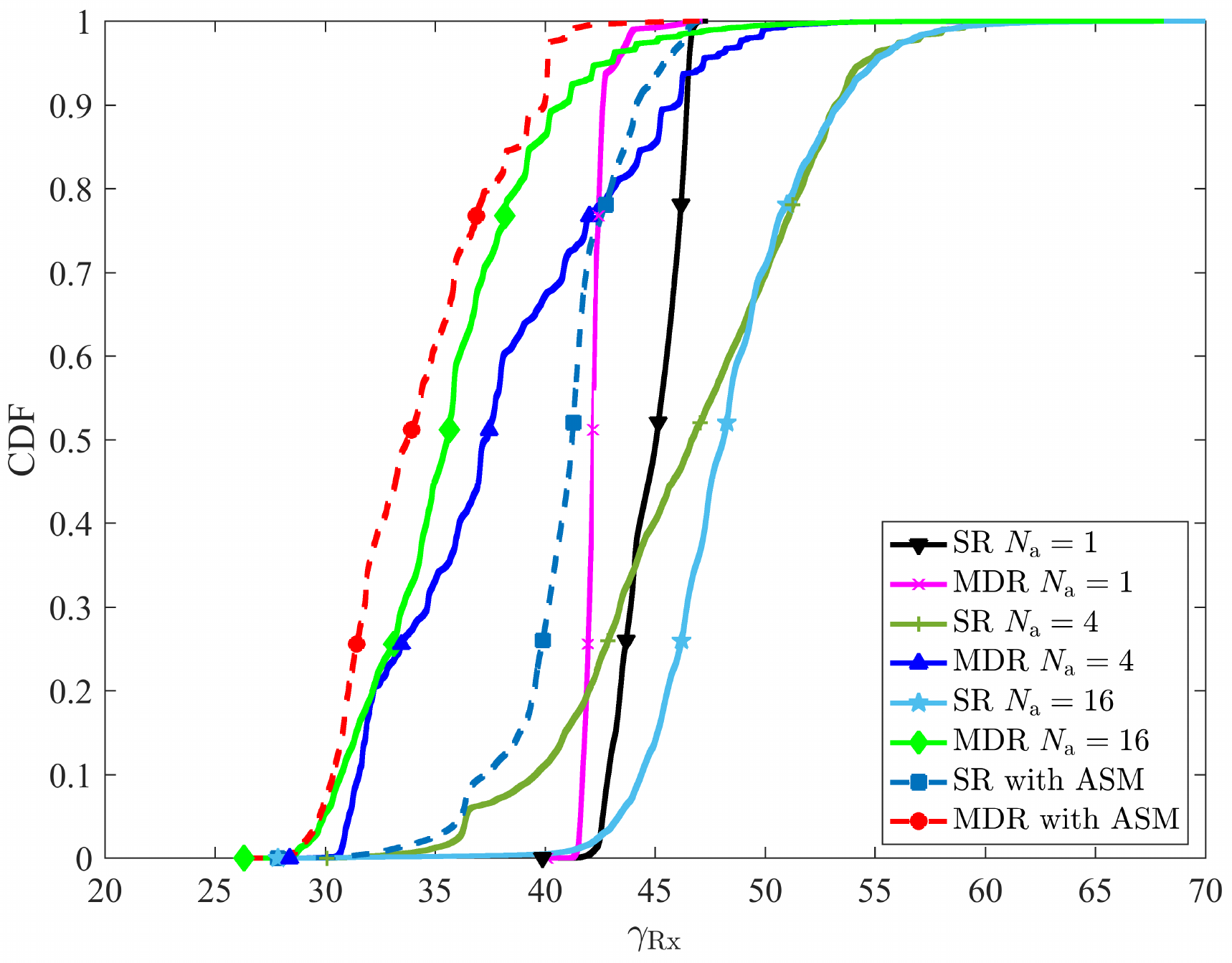}\vspace{-0.2cm}
  \captionof{figure}{Performance comparison of SR and MDR for all UE's locations and directions in the room.}
  \label{CDFSRMDR_AP}
\end{minipage}
\vspace{-0.5cm}
\end{figure}


\vspace{-0.2cm}
\subsection{Mobility}
\label{ORWP}
Performance analysis with consideration of user mobility is crucial in the design of wireless communication networks. The random waypoint (RWP) mobility model is one of the most widely used and simple models, which is utilized for the simulation-based studies of wireless networks \cite{Bettstetter}. The RWP mobility model specifies the following characterizations, i) destinations are chosen randomly following a uniform distribution in the room area, and ii) users move with a constant speed on a straight line between each two consecutive waypoints \cite{Bettstetter}. The RWP model is identified as a discrete-time stochastic random process.  
Mathematically, it can be denoted as an infinite sequence of triples, $ \left\{ \left( \mathbf{P}_{n-1}, \mathbf{P}_n, v_n \right) \big| \ n \in \mathbb{N} \right\} $, where $n$ shows the $n$th movement period. The UE moves from ${\bf{P}}_{n-1}=(x_{n-1},y_{n-1})$ to ${\bf{P}}_{n}=(x_n,y_n)$ with a constant speed of $v$. 
The angle between the direction of movement and the positive direction of the $X$-axis is defined as the user direction, $\Omega=\tan^{-1}(\frac{y_n-y_{n-1}}{x_n-x_{n-1}})$.
		
In order to provide a more realistic framework for analyzing the performance of mobile wireless networks in LiFi, it is required to combine the conventional RWP with the random orientation model. An orientation-based random waypoint (ORWP) mobility model is introduced in \cite{MDSArxiv2018Orientation}, where the elevation angle of the UE is included during user's movement. An altered version of ORWP (where $\alpha, \beta$ and $\gamma$ are encompassed) is used in this study to evaluate the system performance metrics (such as received SNR) more accurately for mobile users. The ORWP can be modeled as an infinite sequence of quadruples, $ \left\{ \left( \mathbf{P}_{n-1}, \mathbf{P}_n, v_n, \mathbf{\Theta}_n \right) \big| \ n \in \mathbb{N} \right\} $, where $\mathbf{\Theta}_{n}=(\alpha_{n},\beta_{n},\gamma_{n})$ is a random vector process describing the UE's orientation during the movement from waypoint ${\bf{P}}_{n-1}$ to waypoint ${\bf{P}}_{n}$. The entities, $\alpha$, $\beta$ and $\gamma$ are random processes (RP). The ORWP is summarized in the Algorithm~\ref{Algorith1}. 
	
	\begin{algorithm}[!t]
		\caption{Orientation-based random waypoint (ORWP)}
		\label{Algorith1}
		
		\begin{algorithmic}[1]
			\STATE Initialization: $n\longleftarrow 1$; $k\longleftarrow 1$;\\ 
			Denote ${\bf{P}}_{n}\!=\!(x_n,y_n)$ and ${\bf{P}}_{0}\!=\!(x_0,y_0)$ as the $n$th and initial UE's positions, respectively;\\
            \vspace{-1mm}
			$N$ as the number of simulation runs; $v$ as the average speed of UE;
			$T_{\rm{c,\alpha}}$, $T_{\rm{c,\beta}}$ and $T_{\rm{c,\gamma}}$ as the coherence time of $\alpha$, $\beta$ and $\gamma$, respectively; Set $T_{\rm{c}}=\min\{T_{\rm{c,\alpha}},T_{\rm{c,\beta}},T_{\rm{c,\gamma}}\}$;\\
			$\mu_{\alpha}$, $\mu_{\beta}$ and $\mu_{\gamma}$ as the mean and $\sigma_{\alpha}^2$, $\sigma_{\beta}^2$ and $\sigma_{\gamma}^2$ as the variance of Gaussian RPs $\alpha$, $\beta$ and $\gamma$;\\
            \vspace{-1mm}
			\FOR{  $n=1:N$ }
			\vspace{-1mm}
			\STATE Choose a random location ${\mathcal{P}}_{n}=(x_n,y_n)$
			\vspace{-1mm}
			\STATE Compute $\mathcal{D}_n=\|{\mathcal{P}}_{n}-\mathcal{P}_{n-1} \|$
			\vspace{-1mm}
			\STATE Compute $\Omega=\tan^{-1}\left(\frac{y_n-y_{n-1}}{x_n-x_{n-1}} \right) $
			\vspace{-1mm}
            \STATE $t_{\rm{move}}\longleftarrow 0$; 
			\vspace{-1mm}
			\WHILE{$t_{\rm{move}}\leq\frac{\mathcal{D}_n}{v}$}
			\vspace{-1mm}
			\STATE Compute ${\bf{P}}_{k}=(x_k,y_k)$ with $x_k=x_{k-1}+vT_{\rm{c}} \cos\Omega$ and $y_k=y_{k-1}+vT_{\rm{c}} \sin\Omega$
			\vspace{-1mm}
			\STATE Generate $\mathbf{\Theta}_k=(\alpha_k,\beta_k,\gamma_k)$ based on the AR(1) model
			\vspace{-1mm}
			\STATE Return $({\bf{P}}_{k-1},{\bf{P}}_{k},v,\mathbf{\Theta}_k)$ as ORWP specifications
            \vspace{-1mm}
			\STATE $k\longleftarrow k+1$
            \vspace{-1mm}
			\STATE $t_{\rm{move}}\longleftarrow t_{\rm{move}}+T_{\rm{c}}$
			\vspace{-1mm}
			\ENDWHILE
			\vspace{-1mm}
			\IF{$t_{\rm{move}}\neq \frac{\mathcal{D}_n}{v}-T_{\rm{c}}$ \& $t_{\rm{move}}\geq\lfloor\frac{\mathcal{D}_n}{v} \rfloor T_{\rm{c}}-T_{\rm{c}}$}
			\vspace{-1mm}
			\STATE Generate $\mathbf{\Theta}_k=(\alpha_k,\beta_k,\gamma_k)$ based on the AR(1) model
			\vspace{-1mm}
			\STATE ${\bf{P}}_{k}\longleftarrow {\mathcal{P}}_{n}$
			\vspace{-1mm}
			\STATE Return $({\bf{P}}_{k-1},{\bf{P}}_{k},v,\mathbf{\Theta}_k)$ as ORWP specifications
            \vspace{-1mm}
			\STATE $k\longleftarrow k+1$
            \vspace{-1mm}
			\ENDIF
            \vspace{-1mm}
			\STATE $n\longleftarrow n+1$
			\vspace{-1mm}
			\ENDFOR
		\end{algorithmic}
	\end{algorithm}

	It is shown in Section~\ref{Sec_rand_ori_blk} that $\alpha$, $\beta$ and $\gamma$ for walking activities follow a Gaussian distribution with the parameters given in Table~\ref{distfit}. According to the experimental measurements, the adjacent samples of the RPs $\alpha$, $\beta$ and $\gamma$  are correlated. Hence, to incorporate the device orientation with the RWP mobility model, a correlated Gaussian RP which statistically follows the experimental measurements should be generated. It should be noted that the random orientation process considered here can be applied to any other mobility models.
Methods to generate a correlated Gaussian RP are discussed in \cite{fox1988fast,CorrelatedGaussTrans} and references therein. A simple way of producing a correlated Gaussian RP is to use a linear time-invariant (LTI) filter and passing a white noise process through it, e.g., a linear autoregressive (AR) filter. 
Thus, after passing the white noise process, $w[k]$, through the LTI filter, the $k$th time sample of the correlated Gaussian RP, $\alpha[k]$, is expressed as:
\begin{equation}
\alpha[k]=c_0+\sum_{i=1}^{p}c_i\alpha[k-i]+w[k],
\end{equation}
where $c_i$ for $i=0,\dots,p$ are constant coefficients of the AR model with order $p$, i.e., AR($p$), and $c_0$ specifies the bias level. To characterize the AR($p$) model, we need to determine $p+2$ unknown parameters that are: $c_0,c_1,\dots,c_p,\sigma_w^2$, and $\sigma_w^2$ is the variance of white noise RP, $w$. 
These parameters can be obtained by matching the generated random process to the moments and the correlation lag between the samples \cite{CorrGausTS}. Here, we use the moments obtained through experimental measurements. Hence, a first-order AR model is sufficient to be considered for generating the correlated Gaussian RP as the mean and variance of the produced samples match the measurement results. The $k$th sample of the AR(1) model is given as:\vspace{-0.2cm}
	\begin{equation}
	\label{AR1}
	\alpha[k]=c_0+c_1\alpha[k-1]+w[k].
	\end{equation}
To guarantee the RP of $\alpha$ is wide-sense stationary, the condition $|c_1|<1$ should be fulfilled.
The mean, variance and autocorrelation function of AR(1) are given respectively as \cite{box2015time}:\vspace{-0.1cm}
	\begin{equation*}
	\label{MeanAR}
	\mu_{\alpha}=\frac{c_0}{1-c_1}, \ \ \ \ \ \ \ \sigma_{\alpha}^2=\frac{\sigma_w^2}{1-c_1^2}, \ \ \ \ \ \ \ \mathcal{R}_{\alpha}(k)=c_1^k.
	\end{equation*}
Note that $\mathcal{R}_{\rm{\alpha}}(\frac{T_{\rm{c,\alpha}}}{T_{\rm{s}}})=0.05$ where $T_{\rm{c,\alpha}}$ is the coherence time of $\alpha$ and $T_{\rm{s}}$ is the sampling time \cite{MDSArxiv2018Orientation}. Using the above equations:\vspace{-0.2cm}
	\begin{equation}
    \label{Parameters}
	c_0=(1-c_1)\mu_{\alpha},\ \ \ \ \ \ \ \sigma_w^2=(1-c_1^2)\sigma_{\alpha}^2, \ \ \ \ \ \ \ 
	c_1=0.05^{\frac{T_{\rm{s}}}{T_{\rm{c,\alpha}}}}. 
	\end{equation}
Then, the $k$th time sample of the correlated Gaussian RP, $\alpha$, can be obtained based on \eqref{AR1} and using the parameters of the AR(1) model given in \eqref{Parameters}. 
The same approach can be applied to both $\beta$ and $\gamma$ to determine the $k$th time sample of the device orientation, $\mathbf{\Theta}_{n}[k]=(\alpha_{n}[k],\beta_{n}[k],\gamma_{n}[k])$. 
According to the approach explained above, the ORWP is described in Algorithm~\ref{Algorith1}.

The simple adaptive algorithm used for the sitting scenario before is also incorporated here in conjunction with the ORWP model. About $500$ random waypoints are generated and the user walks between these points with a constant speed of $1$ m/s. The required received SNRs are calculated along the user's route for $N_\mathrm{a}\in \{1,2,4,8,16\}$ and the adaptive scheme for a target BER of $3.8\times 10^{-3}$ and spectral efficiency of $R=5$ bits/sec/Hz. The results are shown in Fig. \ref{CDFSRMDR_AP1} for $N_\mathrm{a}=1,4,16$ and ASM. Note that ASM is carried out with the choice of $N_\mathrm{a}\in \{1,2,4,8,16\}$ for each channel realization. It can be seen that once again the MDR method significantly improves the performance. The ASM for SR does not change the required SNR value compared to $N_\mathrm{a}=1$, and therefore, $N_\mathrm{a}=1$ is almost optimal for SR in the walking scenario. However, MDR with ASM improves the performance by about $2$ dB gain as compared to a fixed $N_{\rm a}=16$. 

In Fig. \ref{CDFSRMDRMIMO_AP}, the CDF of the required received SNR is simulated for ASM for MDR and SR for two different spectral efficiency values, namely, $R=4$ and $8$ bits/sec/Hz. Moreover, a $4\times4$ full MIMO (i.e., spatial multiplexing) is also considered with both structures. For the full MIMO, the strongest $4$ APs are selected for each channel realization, and the required received SNR is calculated using the union bound method for the BER approximation \cite{fath2013performance}. Interestingly, it is observed that the proposed ASM method outperforms the full MIMO system. As expected, a higher spectral efficiency demands more received SNR in any case. However, by using MDR, $R=8$ bits/sec/Hz can be achieved at similar SNRs compared to SR with $R=4$ bits/sec/Hz at the target BER. This confirms the advantage of using MDR. By increasing the spectral efficiency, the difference between the ASM and full MIMO system reduces. Clearly, the full MIMO system also benefits from the improved channel condition in MDR. Thus, in high spectral efficiencies the full MIMO can be the adopted modulation scheme as it uses spatial degrees of freedom more efficiently \cite{tavakkolnia2018energy}.    

\begin{figure}
\begin{minipage}[t]{.5\textwidth}
  \centering
  \includegraphics[width=1\linewidth]{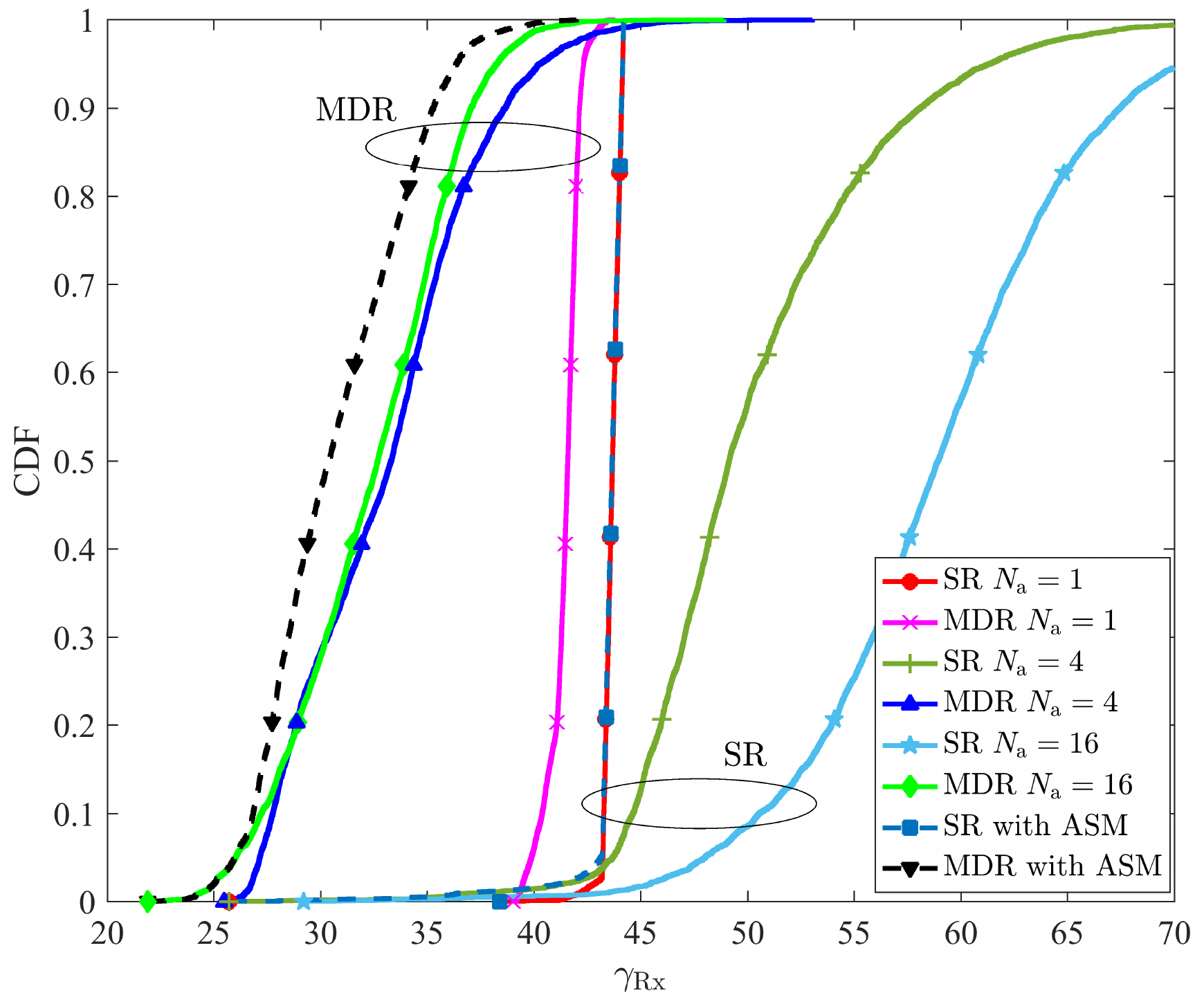}
  \captionof{figure}{Performance comparison of SR and MDR for mobile users based on RWP mobility model.}
  \label{CDFSRMDR_AP1}
\end{minipage}~
\begin{minipage}[t]{.5\textwidth}
  \centering
  \includegraphics[width=1\linewidth]{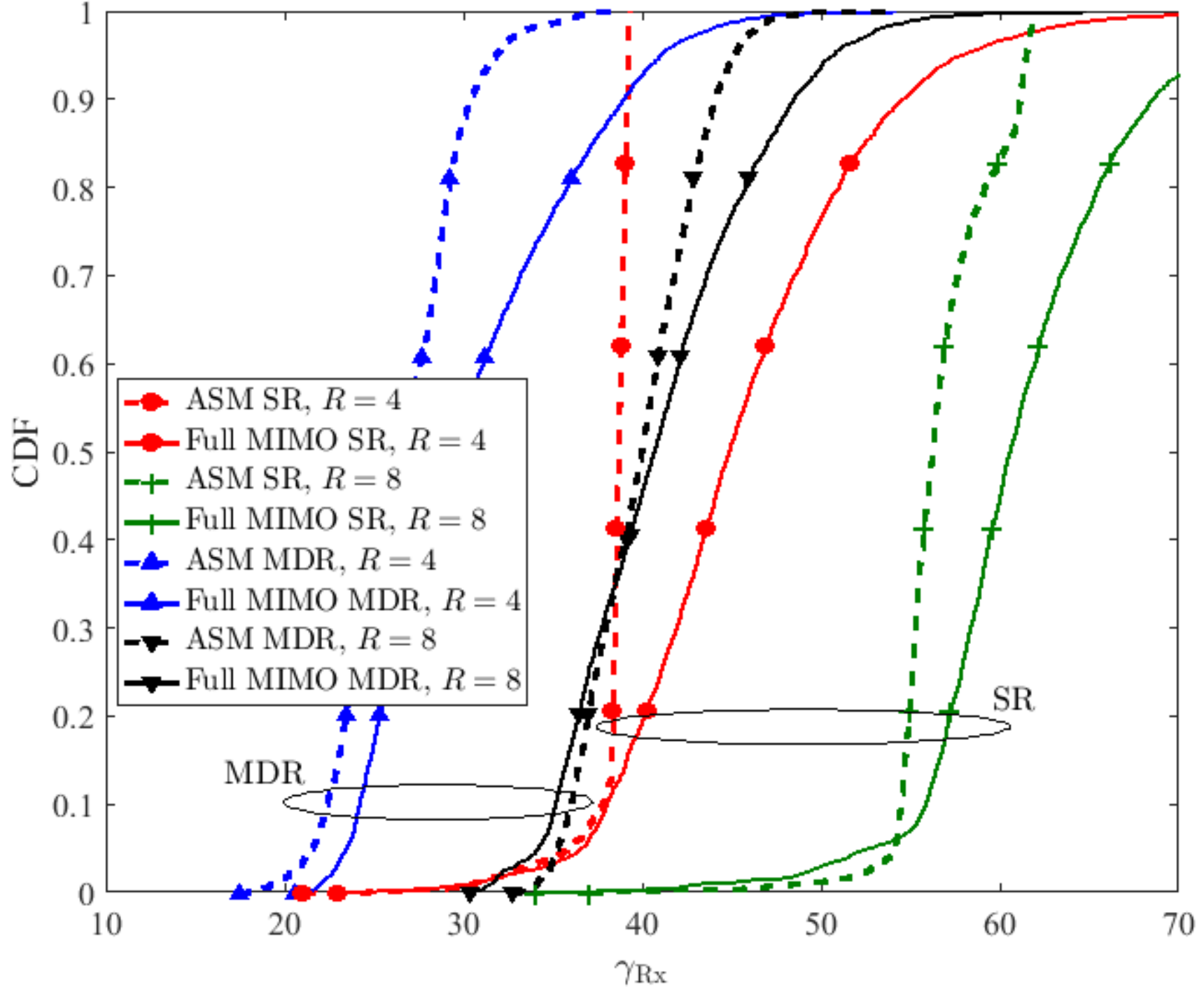}
  \captionof{figure}{Performance comparison of SR and MDR for mobile users with ASM and full MIMO utilization.}
  \label{CDFSRMDRMIMO_AP}
\end{minipage}
\vspace{-0.5cm}
\end{figure}

\section{Performance of the Uplink} \label{sec_uplink}
In indoor bidirectional OWC systems, the uplink transmission presents a fundamental part that needs to be investigated. Similar to downlink transmission, typical uplink performance measures include the average received SNR and the average BER. However, due to the limited transmit uplink power, it is of interest to analyze the energy efficiency in the uplink. \vspace{-0.3cm}
\subsection{Energy Efficiency}  
One important aspect that has to be considered in uplink transmission is the energy efficiency of the transmission scheme. In fact, since the available power at the user's device is not only small, compared to the one at the APs available for downlink transmission, but also limited with the battery reserve, energy efficient transmission schemes present a basic feature when designing mobile devices. In this section, we investigate the energy efficiency of uplink optical SM. For this purpose, we consider the same system model adopted in Section II, where the transmitter is the user's device, that is equipped with $N_{\rm t}=4$ infrared LEDs, and $N_{\rm r}=16$ receivers placed at the AP locations, where each of them is equipped with a single PD. It is important to highlight that, since in the uplink the transmitted power is usually small, the effect of the NLOS channel component can be ignored \cite{ChengOmni}. Therefore, only the LOS channel is considered in this section. \\ 
\indent The energy efficiency $\eta_{\rm EE}$ can be determined through the ratio between the received spectral efficiency $\eta_{\rm RSE}$, in bits per channel use, and the total average transmit power $P$, in Joule per channel use, \cite{ma2018energy,ngo2013energy}. In addition, $P = E_{\rm s} R_{\rm s}$, where $R_{\rm s}$ is the symbol rate. Since at each channel use, only one light source is activated, $P$ is also the average transmit power per light source. Based on this, $\eta_{\rm EE}$ is expressed as:\vspace{-0.3cm}
\begin{equation}
\eta_{\rm EE} = \frac{\eta_{\rm RSE}}{P}.
\end{equation}
Furthermore, based on \cite{ma2018energy,ngo2013energy}, the received spectral efficiency $\eta_{\rm RSE}$ can be determined through the mutual information, measured in bits per channel use, between the transmitted signal $\textbf{x}$ and the received signal $\textbf{y}$. i.e., $\eta_{\rm RSE} = I\left(\textbf{x};\textbf{y}\right)$. However, due to the discreteness of the input signaling, there is no closed form expression for the mutual information $I\left(\textbf{x};\textbf{y}\right)$. As an alternative, we use a modified energy efficiency which is defined as the ratio between the achievable rate $R_{\rm up} \leq I\left(\textbf{x};\textbf{y}\right)$ and the average transmit power $P$, i.e., \vspace{-0.3cm}
\begin{equation}
\eta_{\rm EE} = \frac{R_{\rm up}}{P}.
\end{equation}
In the following theorem, we present an achievable rate $R_{\rm up}$ for the considered optical system. 
\begin{theorem}
An achievable rate $R_{\rm up}$ for the considered optical system is expressed as $R_{\rm up}  = \max \left(L_1^+, L_2^+ \right)$, where the lower bounds $L_1$ and $L_2$ are expressed, respectively, as
\begin{equation}
\left\{ 
\begin{aligned}
&L_1 = 2 \log_2 \left( K \right) - \frac{N_{\rm r}}{2} \left( \log_2 (e) - 1 \right) - \log_2 \left[ \sum_{i=1}^K \sum_{j=1}^K \exp\left(- \frac{1}{\sigma^2} \left| \left| {\bf{H}} \left({\bf{s}}_i - {\bf{s}}_j \right) \right| \right|_2^2 \right) \right], \\ 
&L_2 = 2 \log_2 \left( K \right) +  \frac{1}{2} \log_2 \left[ \frac{|  \frac{2\sigma_x^2}{\sigma^2} {\bf{H}} {\bf{H}}^\mathrm{T} + {\bf{I}}_{N_{\rm r}}|}{| \frac{\sigma_x^2}{\sigma^2} {\bf{H}} {\bf{H}}^\mathrm{T} + {\bf{I}}_{N_{\rm r}}|} \right] - \log_2 \left[ \sum_{i=1}^{K} \sum_{j=1}^{K}  \exp \left( d_{i,j} \right) \right],
\end{aligned}
\right.
\end{equation}
such that $e$ denotes Euler constant, $\sigma_x^2 = \frac{I^2}{3 N_t} \frac{M-1}{M+1}$ is the average emitted power per LED and for all $(i,j) \in \llbracket 1, K \rrbracket^2$, $d_{i,j}$ is expressed as:
\begin{equation}
d_{i,j} = \frac{1}{\sigma^2} {\bf{s}}_i^\mathrm{T} {\bf{H}}^\mathrm{T}\! \left( \frac{2\sigma_x^2}{\sigma^2} {\bf{H}} {\bf{H}}^\mathrm{T} + {\bf{I}}_{N_{\rm r}}\!\! \right)^{-1}\!\! {\bf{H}} {\bf{s}}_j 
- \frac{\sigma_x^2}{2 \sigma^4}  \left({\bf{s}}_i-{\bf{s}}_j \right)^\mathrm{T} {\bf{H}}^\mathrm{T}{\bf{H}} {\bf{H}}^\mathrm{T} \left( \frac{2\sigma_x^2}{\sigma^2} {\bf{H}} {\bf{H}}^\mathrm{T} + {\bf{I}}_{N_{\rm r}}\!\! \right)^{-1}\!\! {\bf{H}} \left({\bf{s}}_i-{\bf{s}}_j \right).
\end{equation}
\end{theorem}
\begin{proof}
See Appendix. 
\end{proof}
As it can be inferred from Theorem 1, the achievable rate $R_{\rm up}$ is the maximum of two different lower bounds $L_1$ and $L_2$. The motivation behind using two lower bounds is explained as follows. The objective is to derive a tight lower bound on the energy efficiency $\eta_{\rm EE}$ which requires the derivation of a tight lower-bound on the received spectral efficiency $\eta_{\rm RSE}$. In other words, the objective is the derivation of an achievable rate that is as close as possible to the transmit spectral efficiency $\eta_{\rm TSE}$. In our case, the lower bounds $L_1$ and $L_2$ depend on the system configuration, i.e., the triplet $(M,N_{\rm t},N_{\rm r})$. This dependency can be clearly seen in the high SNR regime. In fact, for the case where $\frac{I^2}{\sigma^2} \rightarrow \infty$, the lower bounds $L_1$ and $L_2$ can be expressed, respectively, as $L_1 = \eta_{\rm TSE} - \Delta_1$ and $L_2 = \eta_{\rm TSE} - \Delta_2$ where the gaps $\Delta_1$ and $\Delta_2$ are expressed, respectively, as:\vspace{-0.1cm}
\begin{equation}
\left\{ 
\begin{aligned}
&\Delta_1 \approx \frac{N_{\rm r}}{2} \left( \log (e) - 1 \right), \\ 
&\Delta_2 \approx \log_2 \left[ \sum_{i=1}^{M}  \exp \left( \frac{3N_{\rm t}\left(2i-M-1\right)}{2\left(M^2-1\right)}  \right) \right]  -\log_2 \left( M \right) -  \frac{1}{2},
\end{aligned}
\right.
\end{equation}
and consequently, the achievable rate $R_{\rm up}$ is expressed at high SNRs as $R_{\rm up} = \eta_{\rm TSE} - \Delta$, where $\Delta = \min \left( \Delta_1, \Delta_2 \right)$. For instance, when $(M,N_{\rm t},N_{\rm r}) = (4,4,16)$, $\Delta_1 = 3.5416$, $\Delta_2 = 0.0319$ and $\Delta = 0.0319$ whereas when $(M,N_{\rm t},N_{\rm r}) = (4,16,4)$, $\Delta_1 = 0.8854$, $\Delta_2 = 4.4850$ and $\Delta = 0.8854$. This example shows that, when one lower bound is lossy, the other bound compensates for this loss. Note that lower bound $L_1$ is more suitable for downlink and lower bound $L_2$ is more suitable for uplink. In addition, based on its structure and the above example, the achievable rate $R_{\rm up}$ is tight, which implies the tightness of the energy efficiency $\eta_{\rm EE}$.\vspace{-0.3cm}
\subsection{Adaptive Spatial Modulation for Uplink}
Due to the mobility and random orientation of the UE, and in conjunction with the blockage effect within the room, the channel gain varies over time and location. This implies that, for a given channel use, one or more transmit LEDs are inactive, which highly affects the performance of the system. Hence, adaptive modulation is necessary to guarantee reliable and energy efficient transmissions. Here, we propose using ASM (similar to the downlink case), which is based on LEDs selection. 

Assuming that the transmitter knows perfectly the channel matrix $\textbf{H}$ and the noise variance $\sigma^2$ at the receiver, one can  apply a certain selection criterion to estimate the number and the indexes of the ``active" transmit LEDs, which we denote by $N_a$. It is important to highlight that the performance of the system is directly dependent on the choice of the LEDs selection criterion. A simple selection criterion can be  based on the channel matrix $\textbf{H}$. It may consider a LED as active if the total sum of the magnitudes of the channel gains between the considered LED and the received PDs is higher than a fixed threshold $\delta$. In other words, for all $i \in \llbracket 1,N_{\rm t} \rrbracket$, the $i$th LED is active if $\delta \leq ||\textbf{h}_i||_\infty$, where $\textbf{h}_i$ denotes the $i$th column of the channel matrix $\textbf{H}$. However, in order to achieve better performance, the selection criterion should take into account the target performances, i.e., the target spectral efficiency $\eta_{\rm TSE}$ and the target BER $P_{\rm e,tr}$. \\ 
\indent In this paper, we consider the following selection criterion. We assume that the $N_{\rm t}$ LEDs are operating independently, where each of them should achieve the target spectral efficiency $\eta_{\rm TSE}$. In this case, the PAM order per light source is $M=2^{\eta_{\rm TSE}}$. A light source is active if its achieved BER is lower than the target BER $P_{\rm e,tr}$, i.e. $\mathrm{BER}\left(M, E_{\rm s}, {\textbf{h}}_{i} \right) \leq P_{\rm e,tr}$. Based on this criterion, our proposed LEDs selection algorithm is presented in the following. First, we order the column vectors of the matrix $\textbf{H}$ in an ascending order and the resulting matrix is denoted by $\tilde{\textbf{H}} = \left[ \tilde{\textbf{h}}_1, \tilde{\textbf{h}}_2,...,\tilde{\textbf{h}}_{N_{\rm t}} \right]$, so that, 
$||\tilde{\textbf{h}}_1||_2 \leq ||\tilde{\textbf{h}}_2||_2 \leq ... \leq ||\tilde{\textbf{h}}_{N_{\rm t}}||_2$. 
Let $f$ be the integer function that defines the index of the transmit light source after ascending order, i.e., for all $i \in \llbracket 1,N_{\rm t} \rrbracket$, $f(i)$ denotes the ordered index of $i$. In this case, for all $ i \in \llbracket 1,N_{\rm t} \rrbracket$, $\textbf{h}_i = \tilde{\textbf{h}}_{f(i)}$. Moreover, let $\mathcal{G}$ be the set of admissible numbers of active light sources, i.e., $\mathcal{G} \triangleq \left\{ i \in \llbracket 1,N_{\rm t} \rrbracket \left| \right. \, \, \log_2 \left( N_{\rm t} - i + 1 \right) \in \mathbb{N} \right\}$. In addition, let $r$ be the cardinality of $\mathcal{G}$ and let $\textbf{g} = \left[g_1, g_2,...,g_r \right]$ be the $1 \times r$ vector that contains all elements of $\mathcal{G}$ but in an ascending order. \\
The LEDs selection algorithm is detailed as follows. We start by $g_1$ and we consider the channel vector $\tilde{\textbf{h}}_{g_1}$. If $\mathrm{BER}\left(M, E_{\rm s}, \tilde{\textbf{h}}_{g_1} \right) \leq P_{\rm e,tr}$, then there is $N_{\rm a} = N_{\rm t} - g_1 + 1$ active LEDs, where their indexes are given by $\mathcal{E} = \left\{ \left. f^{-1}(j) \right| \, \, j \in \llbracket g_1,N_{\rm t} \rrbracket  \right\}$. Otherwise, we redo the same test to the index $g_2$. Thus, if $\mathrm{BER}\left(M, E_{\rm s}, \tilde{\textbf{h}}_{g_2} \right) \leq P_{\rm e,tr}$, then there is $N_{\rm a} = N_{\rm t} - g_2 + 1$ active LEDs, where their indexes are given by $\mathcal{E} = \left\{ \left. f^{-1}(j) \right| \, \, j \in \llbracket g_2,N_{\rm t} \rrbracket  \right\}$. Otherwise, we continue until either we reach a non-null number of active LEDs $N_{\rm a}$ or when all the LEDs are inactive, i.e., $N_{\rm a}=0$. In the latter case, the communication fails and the user should be notified to change its orientation or location. Based on the above, the detailed LEDs selection algorithm is given in Algorithm \ref{Algorith2}.\vspace{-0.3cm}
\begin{algorithm}[!t]
		\caption{LEDs Selection Algorithm}
		\label{Algorith2}		
		\begin{algorithmic}[1]
        	\STATE Input: $\textbf{H}$; \\
            \vspace{-1mm}
            \STATE Construct: $\tilde{\textbf{H}}$ and $\textbf{g}$; \\
            \vspace{-1mm}
			\STATE Initialization: $N_{\rm a} \longleftarrow 0$, $\mathcal{E} \longleftarrow \left\{ \right\}$, $state \longleftarrow 0$, $i\longleftarrow 1$ and $r \longleftarrow \rm{length} \left(\textbf{g} \right) $;\\ 
            \vspace{-1mm}
			\WHILE{$state = 0$ and $i \leq r$}
			\vspace{-1mm}
            \IF{$\mathrm{BER}\left(M, E_{\rm s}, \tilde{\textbf{h}}_{g_i} \right) \leq P_{\rm e,tr}$}
			\vspace{-1mm}
			\STATE $N_{\rm a} \longleftarrow N_{\rm t} - g_i + 1$
			\vspace{-1mm}
			\STATE $\mathcal{E} \longleftarrow \left\{ \left. f^{-1}(j) \right| \, \, j \in \llbracket g_i,N_{\rm t} \rrbracket  \right\}$
			\vspace{-1mm}
			\STATE $state \longleftarrow 1$
            \vspace{-1mm}
			\ENDIF
			\vspace{-1mm}
            \STATE $i\longleftarrow i+1$
            \vspace{-1mm}
			\ENDWHILE
            \vspace{-1mm}
            \IF{$N_{\rm a} = 0$}
			\vspace{-1mm}
			\STATE Communication failed. User should change its orientation or location.
			\ENDIF
		\end{algorithmic}
	\end{algorithm}
    \vspace{-0.3cm}
\subsection{Simulation Results}
In order to validate the proposed modulation schemes, we consider the same indoor environment adopted for the downlink. The uplink system consists of a UE, possibly with random orientation, that is equipped with $N_{\rm t} = 4$ infrared LEDs, and $N_{\rm r} = 16$ receivers placed at the AP locations on the ceiling of the room. Furthermore, each receiver is equipped with a facing-down single PD. Two configurations for the UE are considered, namely, ST and MDT as shown in Fig.~\ref{smartphone1} and Fig.~\ref{smartphone2}, respectively, where the infrared LEDs are located beside the PDs.


\begin{figure}[t]
		\centering
		\begin{subfigure}[b]{0.5\columnwidth}
			\centering
			\includegraphics[width=1\columnwidth,draft=false]{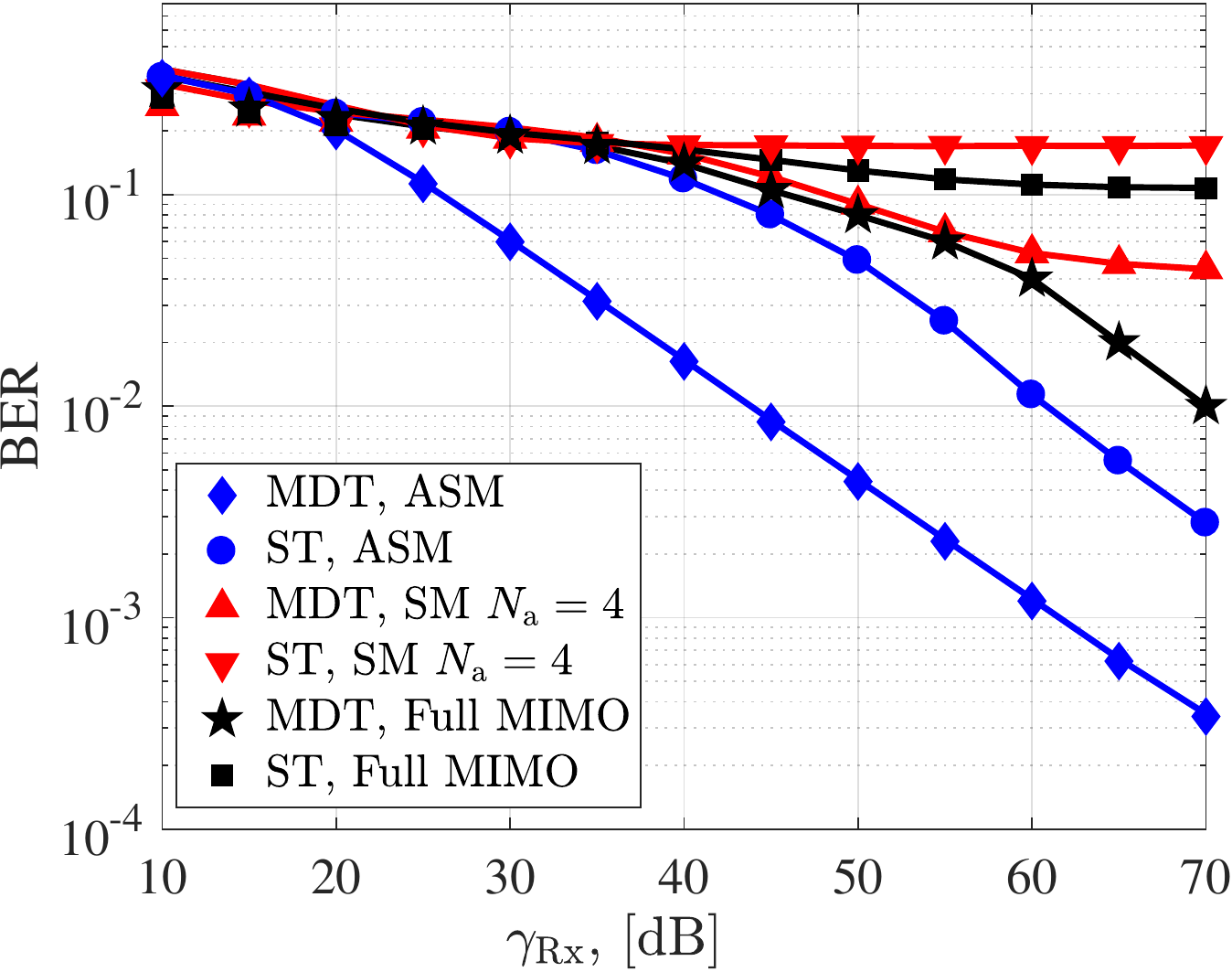}
						\caption{Sitting scenario.}
				\label{BERup2}
                \vspace{-0.3cm}
		\end{subfigure}%
		~
		\begin{subfigure}[b]{0.5\columnwidth}
			\centering
			\includegraphics[width=1\columnwidth,draft=false]{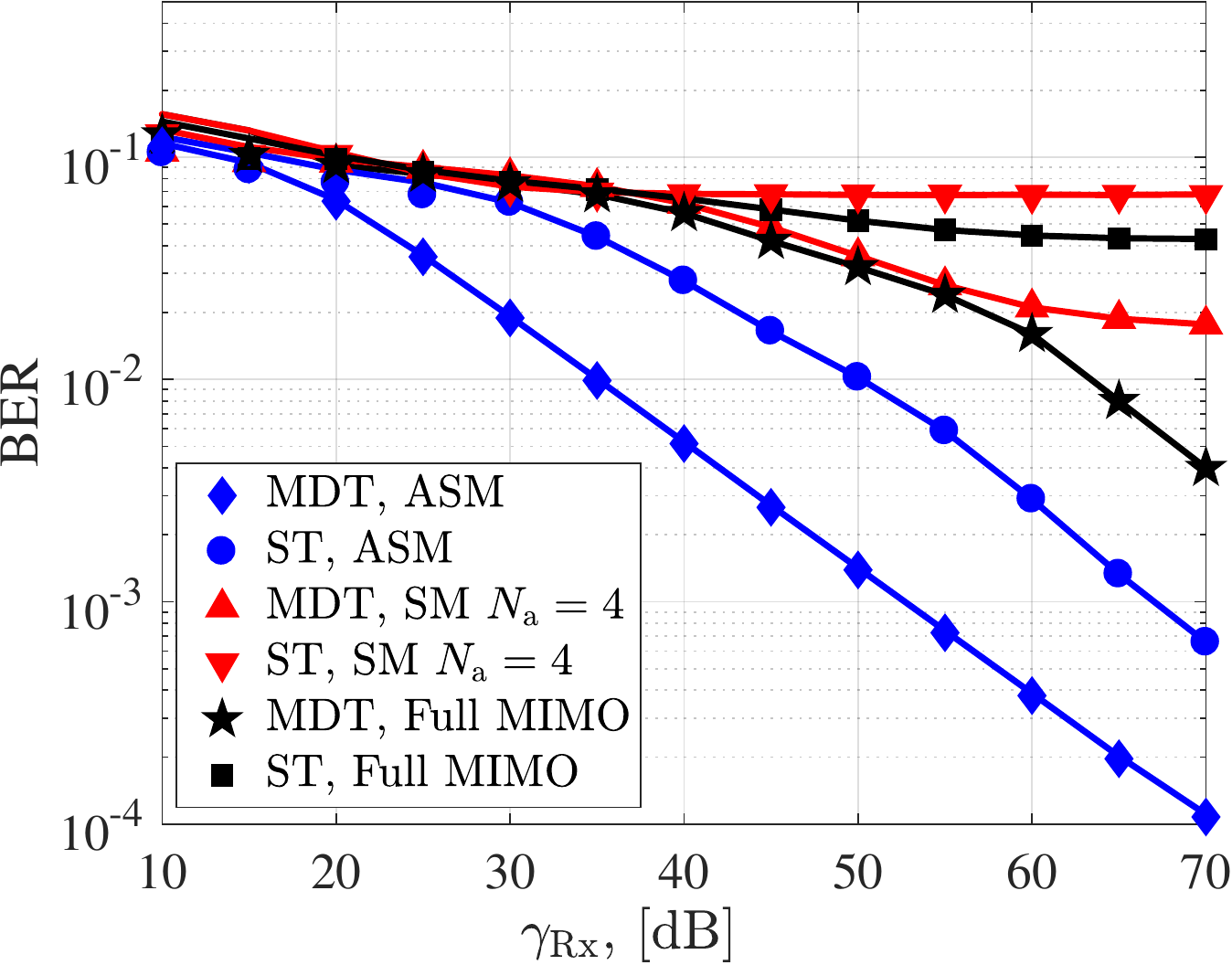}
						\caption{Walking scenario.}
				\label{BERup1}
                \vspace{-0.3cm}
		\end{subfigure}
				\caption{BER versus average received SNR $\gamma_{\rm Rx}$. Marks denote Monte-Carlo simulations and solid lines denotes theoretical results.}
				\label{BERup}
                \vspace{-0.6cm}
\end{figure}

\subsubsection*{BER Evaluation} Figs.~\ref{BERup2} and \ref{BERup1} present the average BER versus the average received SNR $\gamma_{\rm Rx}$ for the sitting and walking scenarios, respectively. For the sitting scenario, the average BER is obtained over all room locations, whereas for the walking scenario, the ORWP described in Algorithm~\ref{Algorith1} is considered. Results are demonstrated for the MDT and ST configurations using the ASM, SM and full MIMO schemes. Both figures show that MDT outperforms ST. As can be seen, for a BER of $3.8\times10^{-3}$ and when ASM is employed, there is a gap of approximately $15$ dB between MDT and ST configurations for both the sitting and walking scenarios. On the other hand, Figs.~\ref{BERup2} and ~\ref{BERup1} show that ASM significantly outperforms SM. Finally, by comparing both figures, we note that the average BER of mobile users is lower than that of static users. This is attributed to the fact that, when the user is moving, the height of the UE is higher compared to the sitting case. Therefore, for the walking scenario, the user's device is closer to the APs than the case in the sitting scenario, which increases the average received SNR. The other interesting observation refers to the notable gap between the proposed ASM algorithm, SM and full MIMO.   

\subsubsection*{Energy Efficiency Evaluation} The average received energy efficiency $\eta_{\rm EE}$ versus the average received spectral efficiency $\eta_{\rm RSE}$ are presented in Figs.~\ref{EE2} and \ref{EE1} for the sitting and walking scenarios, respectively. Performance of both MDT and ST using ASM, SM and full MIMO schemes are compared. These figures show that MDT outperforms ST. For example, for the walking scenario and when $\eta_{\rm RSE}=3$ bits/s/Hz, we can see that MDT is approximately $4$ times more energy efficient than ST. Whereas when $\eta_{\rm RSE}=2$ bits/s/Hz and SM is employed, it can be seen that MDT is approximately $1.5$ times more energy efficient than ST. On the other hand, Figs.~\ref{EE2} and ~\ref{EE1} show that ASM outperforms both SM and full MIMO. As an illustration, we can see that for the sitting scenario and when $\eta_{\rm RSE}=1.5$ bits/s/Hz, there is a gap of approximately $70$ bits/J between ASM and SM when MDT is adopted, whereas for ST, the gap is approximately $20$ bits/J. Finally, by comparing both figures and similar to the spectral efficiency evaluation, we note that the resulting average energy efficiency of mobile users is higher than the one of static user. \vspace{-0.3cm}
\begin{figure}[t]
		\centering
		\begin{subfigure}[b]{0.5\columnwidth}
			\centering
			\includegraphics[width=1\columnwidth,draft=false]{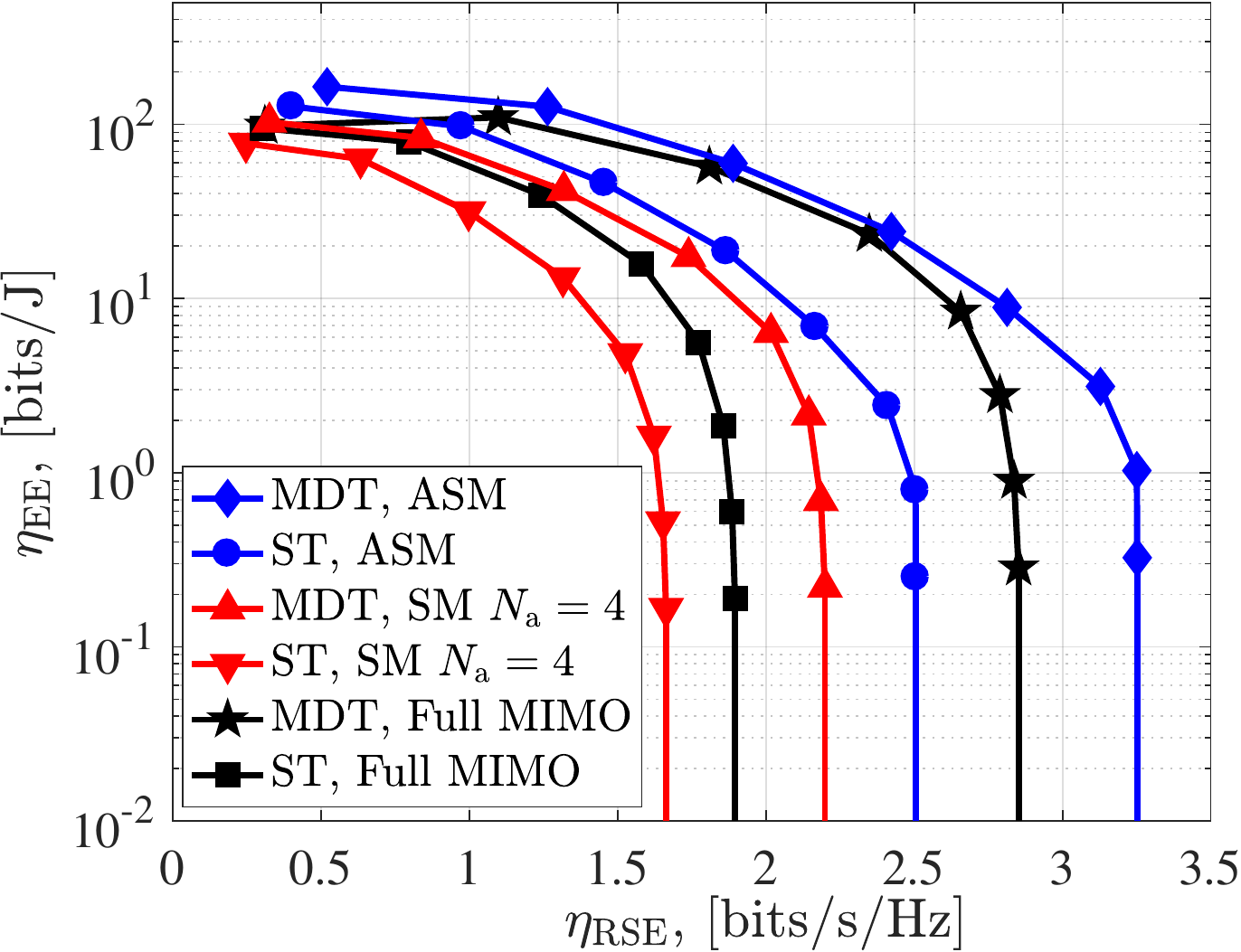}
						\caption{Sitting scenario.}
				\label{EE2}
                \vspace{-0.3cm}
		\end{subfigure}%
		~
		\begin{subfigure}[b]{0.5\columnwidth}
			\centering
			\includegraphics[width=1\columnwidth,draft=false]{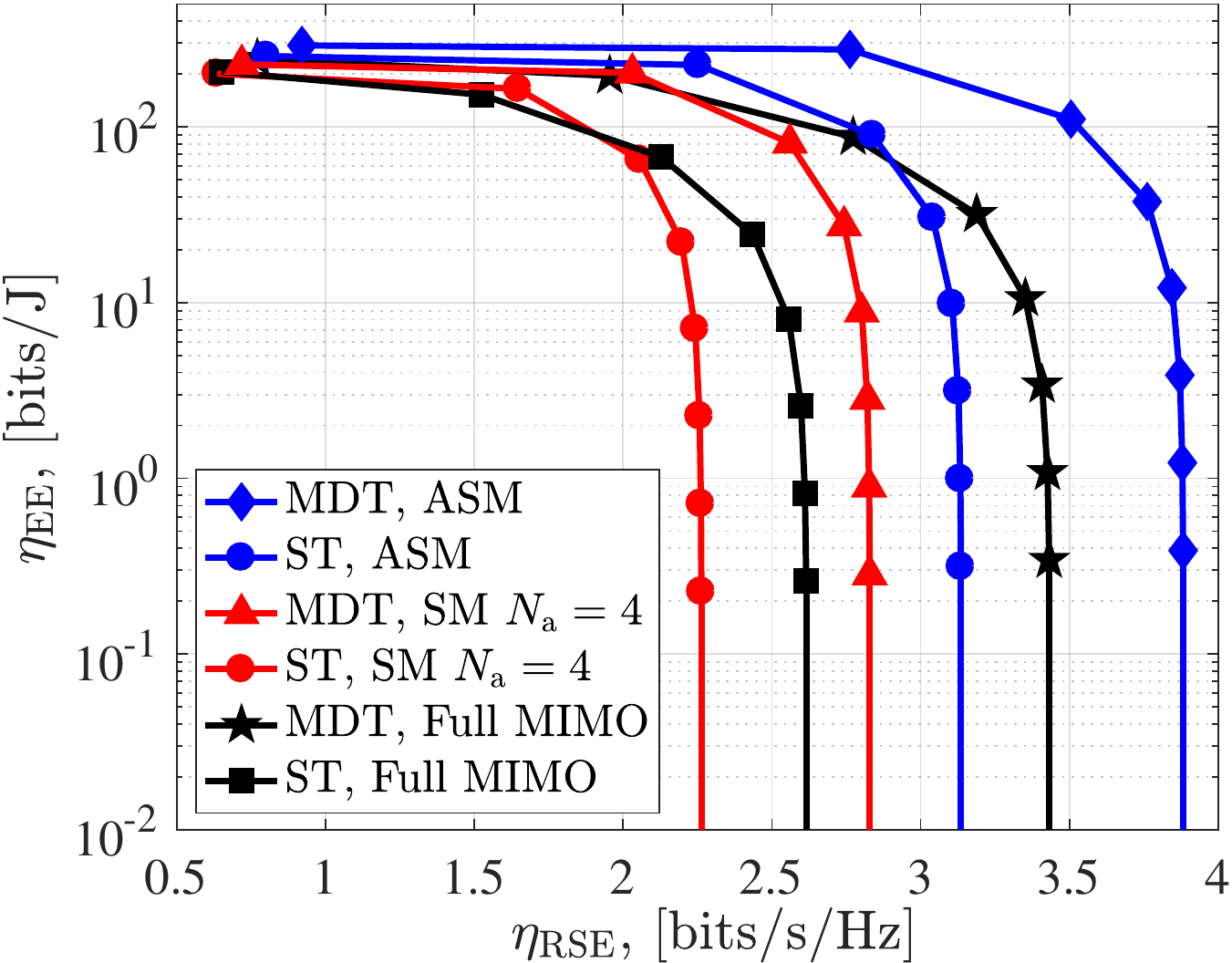}
						\caption{Walking scenario.}
				\label{EE1}
                \vspace{-0.3cm}
		\end{subfigure}
				\caption{Average energy efficiency $\eta_{\rm EE}$ versus received spectral efficiency $\eta_{\rm RSE}$.}
				\label{EE}
                \vspace{-0.6cm}
\end{figure}
\section{Conclusions and Future Work} \label{sec_conc}
In this paper, the effects of mobility, random orientation and blockage on indoor bidirectional optical SM were investigated while adopting a channel model derived from real-life measurements. A new user device configuration, called MDR in downlink and MDT in uplink were proposed to overcome the problem of high channel correlation. In addition, an adaptive SM scheme based on AP selection for downlink (and PDs for uplink) was proposed to overcome the effect of random orientation and blockage and to reduce power consumption. It was shown that MDR/MDT provide superior performance in terms of SNR, BER and energy efficiency and significantly outperform their counterpart SR/ST configurations. Since the proposed schemes were analyzed separately, a future extension of this work could be a joint optimization of the uplink/downlink, especially if the channel gains are available at the transmitter. On the other hand, recall that the performance of SM in this paper was investigated for an indoor communication system. Therefore, the validity of the obtained results for outdoor environments should also be considered in future work. Moreover, incorporating the multiuser scenario and SM will be another direction that we will explore in a future study.  

\appendix 
\section*{Proof of Theorem 1}
Based on \cite{cover2006elements}, the mutual information $I(\textbf{x};\textbf{y})$ can be expressed as:
\begin{equation}
I(\textbf{x};\textbf{y}) = h (\textbf{y}) - h ( \textbf{y} | \textbf{x} ),
\end{equation}
where $h ( \textbf{y} | \textbf{x} ) = h ( \textbf{n} ) = \frac{N_{\rm r}}{2} \log \left(2 \pi e \sigma^2  \right) $. In addition, since $\textbf{y} = \textbf{H}  \textbf{x} + \textbf{n}$ and $p_{\mathcal{X}} \left( \textbf{x} \right) = \frac{1}{K} \sum_{i=1}^K \delta \left(\textbf{x} - {\bf{s}}_i \right)$, the PDF of the received signal $\textbf{y}$ is a mixture of multivariate Gaussian distribution that is expressed, for all $\textbf{z} \in \mathbb{R}^{N_{\rm r}}$ as $p_\mathcal{Y} (\textbf{z}) = \frac{1}{K} \sum_{i=1}^K \mathcal{N} \left(\textbf{z} \left| \textbf{H}{\bf{s}}_i \right., \sigma^2 \textbf{I}_{N_{\rm r}} \right)$, where for all $\boldsymbol \mu \in \mathbb{R}^{N_{\rm r}}$ and for all $N_{\rm r} \times N_{\rm r}$ invertible matrix $\textbf{G}$, 
\begin{equation}
\mathcal{N} \left(\textbf{z} \left| \boldsymbol \mu \right., \textbf{G} \right) \triangleq \frac{1}{\sqrt{\left(2 \pi \right)^{N_{\rm r}}\left| \textbf{G} \right|}} \exp \left( - \frac{\left(\textbf{z} - \boldsymbol \mu \right)^\mathrm{T}\textbf{G}^{-1}\left( \textbf{z} - \boldsymbol \mu \right) }{2} \right).
\end{equation}
However, since there is no closed-form expression for the differential entropy of a mixture of multivariate Gaussian distribution, the derivation of $h\left( \textbf{y} \right)$ is not straightforward. Therefore, we provide in the following two lower bounds on the differential entropy $h\left( \textbf{y} \right)$, which lead to the lower bounds $L_1$ and $L_2$ provided in the theorem 1.
\subsection{Lower bound {$L_1$}}
\label{App1}
The differential entropy $h\left( \textbf{y} \right)$ can be lower bounded as:
\begin{subequations}
\begin{align}
h\left( \textbf{y} \right) &= - \int_{\mathbb{R}^{N_{\rm r}}} p_\mathcal{Y} (\textbf{z}) \log_2 \left[ p_\mathcal{Y} (\textbf{z}) \right] \mathrm{d} \textbf{z}\\ 
&\geq - \log_2 \left[ \int_{\mathbb{R}^{N_{\rm r}}} p_\mathcal{Y} (\textbf{z})^2 \mathrm{d} \textbf{z} \right] \\ 
&= -\log_2 \left[ \frac{1}{K^2}  \sum_{i=1}^K \sum_{j=1}^K \int_{\mathbb{R}^{N_{\rm r}}} \mathcal{N} \left(\textbf{z} \left| \textbf{H}{\bf{s}}_i \right., \sigma^2 \textbf{I}_{N_{\rm r}} \right) \mathcal{N} \left(\textbf{z} \left| \textbf{H}{\bf{s}}_j \right., \sigma^2 \textbf{I}_{N_{\rm r}} \right) \mathrm{d} \textbf{z} \right] \\ 
&= -\log_2 \left[ \frac{1}{K^2}  \sum_{i=1}^K \sum_{j=1}^K \frac{\exp \left( -\frac{1}{4 \sigma^2} \left| \left| \textbf{H} \left( {\bf{s}}_i - {\bf{s}}_j \right) \right| \right|_2^2 \right)}{\left(4 \pi \sigma^2 \right)^{\frac{N_{\rm r}}{2}}} \underbrace{\int_{\mathbb{R}^{N_{\rm r}}} \mathcal{N} \left(\textbf{z} \left| \textbf{H} \frac{\left({\bf{s}}_i + {\bf{s}}_j \right)}{2}  \right., \frac{\sigma^2}{2} \textbf{I}_{N_{\rm r}} \right) \mathrm{d} \textbf{z} }_{=1} \right] \\ 
&= \frac{N_{\rm r}}{2} \log_2 \left(2 \pi e \sigma^2  \right)- \frac{N_{\rm r}}{2} \left( \log_2 \left(e \right) - 1 \right) - \log \left[ \frac{1}{K^2} \sum_{i=1}^K \sum_{j=1}^K \exp \left( -\frac{1}{4 \sigma^2} \left| \left| \textbf{H} \left( {\bf{s}}_i - {\bf{s}}_j \right) \right| \right|_2^2 \right) \right],
\end{align}
\end{subequations}
where inequality (23b) follows from the concavity of the logarithmic function and Jensen's inequality. Consequently, by substituting each term in (22) by its expression, the mutual information $I(\textbf{x};\textbf{y})$ can be lower bounded as $ L_1^+ \leq I(\textbf{x};\textbf{y})$, where 
\begin{equation}
L_1 = 2 \log_2 \left[ K \right] - \frac{N_r}{2} \left( \log_2 \left(e \right) - 1 \right) - \log \left[ \sum_{i=1}^K \sum_{j=1}^K \exp \left( -\frac{1}{4 \sigma^2} \left| \left| \textbf{H} \left( {\bf{s}}_i - {\bf{s}}_j \right) \right| \right|_2^2 \right) \right].
\end{equation}
\subsection{Lower bound {$L_2$}}
\label{App2}
Let $\textbf{K}_\textbf{x}$ be the covariance matrix of the transmitted signal $\textbf{x}$. Consider the Gaussian vectors $\textbf{x}^C$ and $\textbf{y}^C$, where $\textbf{x}$ follows $\mathcal{N} \left(\textbf{0}_{N_{\rm r}}, \textbf{K}_\textbf{x} \right)$ and $\textbf{y}^C = \textbf{x}^C + \textbf{n}$. In this case, $\textbf{y}^c$ follows $\mathcal{N} \left(\textbf{0}_{N_{\rm r}}, \textbf{K}_c \right)$, where $\textbf{K}_c = \textbf{H} \textbf{K}_\textbf{x} \textbf{H}^\mathrm{T} + \sigma^2 \textbf{I}_{N_{\rm r}}$. Since $\textbf{y}$ and $\textbf{y}^C$ have the same covariance matrix, which is $\textbf{K}_c$, and based on \cite{cover2006elements}, we have: \vspace{-0.2cm}
\begin{equation}
h(\textbf{y}^C) - h(\textbf{y}) = \mathcal{D} \left(\textbf{y}||\textbf{y}^C \right) = \mathbb{E}_{\mathcal{Y}} \left( \log \left[ \frac{p_{\mathcal{Y}}}{p_{\mathcal{Y}^C}} \right] \right),
\end{equation}
where $\mathcal{D}$ denotes the Kullback--Leibler divergence. On the other hand, using the concavity of the logarithm function and Jensen's inequality, we have $
\mathbb{E}_{\mathcal{Y}} \left( \log_2 \left[ \frac{p_{\mathcal{Y}}}{p_{\mathcal{Y}^C}} \right] \right) \leq \log_2 \left[ \mathbb{E}_{\mathcal{Y}} \left( \frac{p_{\mathcal{Y}}}{p_{\mathcal{Y}^C}} \right) \right]$, and consequently, $h(\textbf{y})$ can be lower bounded as: 
\begin{equation}
h(\textbf{y}) \geq h(\textbf{y}^C) - \log_2 \left[ \mathbb{E}_{\mathcal{Y}} \left( \frac{p_{\mathcal{Y}}}{p_{\mathcal{Y}^C}} \right) \right].
\end{equation}
Furthermore, $ \mathbb{E}_{\mathcal{Y}} \left( \frac{p_{\mathcal{Y}}}{p_{\mathcal{Y}^C}} \right) $ is computed as:
\begin{equation}
\begin{split}
\mathbb{E}_{\mathcal{Y}} \left( \frac{p_{\mathcal{Y}}}{p_{\mathcal{Y}^C}} \right) &= \int_{\mathbb{R}^{N_{\rm r}}} p_{\mathcal{Y}}(\textbf{z}) \frac{p_{\mathcal{Y}}(\textbf{z})}{p_{\mathcal{Y}^C}(\textbf{z})} \mathrm{d}\textbf{z}  
= \sum_{i=1}^{K} \sum_{j=1}^{K} \frac{1}{K^2} \int_{\mathbb{R}^{N_{\rm r}} } \frac{\mathcal{N} \left(\textbf{z} \left|\textbf{H}{\bf{s}}_i \right., \sigma^2 \textbf{I}_{N_{\rm r}} \right) \mathcal{N} \left(\textbf{z} \left| \textbf{H}{\bf{s}}_j \right., \sigma^2 \textbf{I}_{N_{\rm r}} \right)}{\mathcal{N} \left(\textbf{z} \left| \textbf{0}_{N_{\rm r}} \right., \textbf{K}_c \right)} \\
&= \frac{\left(2 \pi \right)^{\frac{{N_{\rm r}}}{2}} |\textbf{K}_c|^{\frac{1}{2}}}{\left( 2 \pi \sigma^2 \right)^{N_{\rm r}} } \sum_{i=1}^{K} \sum_{j=1}^{K} \frac{1}{K^2} \exp \left[ d_{i,j} \right] \int_{\mathbb{R}^{N_{\rm r}} } \exp \left[-\frac{ \left(\textbf{y} - \textbf{u}_{i,j} \right)^\mathrm{T} \textbf{A}^{-1} \left(\textbf{y} - \textbf{u}_{i,j} \right)  }{2} \right] \mathrm{d}\textbf{y} \\ 
&= \frac{\left(2 \pi e \right)^{\frac{{N_{\rm r}}}{2}} |\textbf{K}_c|^{\frac{1}{2}} \left( 2 \pi e \right)^\frac{{N_{\rm r}}}{2} |\textbf{A} |^{\frac{1}{2}} }{\left( 2 \pi e \sigma^2 \right)^{N_{\rm r}}  } \sum_{i=1}^{K} \sum_{j=1}^{K} \frac{1}{K^2} \exp \left[ d_{i,j} \right] \\ 
&= \frac{\left(2 \pi e \right)^{\frac{{N_{\rm r}}}{2}} |\textbf{K}_c|^{\frac{1}{2}}  }{\left( 2 \pi e \sigma^2 \right)^\frac{{N_{\rm r}}}{2} }   \frac{| \frac{1}{\sigma^2} \textbf{H} \textbf{K}_\textbf{x}\textbf{H}^\mathrm{T} + \textbf{I}_{N_{\rm r}}|^{\frac{1}{2}}}{|  \frac{2}{\sigma^2} \textbf{H} \textbf{K}_\textbf{x}\textbf{H}^\mathrm{T} + \textbf{I}_{N_{\rm r}}|^{\frac{1}{2}}}   \sum_{i=1}^{K} \sum_{j=1}^{K} \frac{1}{K^2} \exp \left[ d_{i,j} \right],
\end{split}
\end{equation} 
where, 
\begin{equation} 
\left\{ 
\begin{aligned} 
&\textbf{A} = \sigma^2 \textbf{K}_c \left(2 \textbf{K}_c - \sigma^2\textbf{I}_{N_{\rm r}} \right)^{-1} = \sigma^2 \left( \frac{1}{\sigma^2} \textbf{H} \textbf{K}_\textbf{x}\textbf{H}^\mathrm{T} + \textbf{I}_{N_{\rm r}} \right) \left(  \frac{2}{\sigma^2} \textbf{H} \textbf{K}_\textbf{x}\textbf{H}^\mathrm{T} + \textbf{I}_{N_{\rm r}} \right)^{-1}, \\ 
&\textbf{u}_{i,j} = \textbf{A}^{-1} \frac{\textbf{H}\left( {\bf{s}}_i + {\bf{s}}_j \right)}{\sigma^2},\\
&d_{i,j} = \frac{1}{\sigma^2} {\bf{s}}_i^\mathrm{T} \textbf{H}^\mathrm{T} \left( \frac{2}{\sigma^2} \textbf{H} \textbf{K}_\textbf{x}\textbf{H}^\mathrm{T} + \textbf{I}_{N_{\rm r}} \right)^{-1} \textbf{H} {\bf{s}}_j \\ 
&\quad \, \, -\frac{1}{2} \left( \frac{1}{\sigma^2} \right)^2 \left({\bf{s}}_i-{\bf{s}}_j \right)^\mathrm{T} \textbf{H}^\mathrm{T}\textbf{H} \textbf{K}_\textbf{x}\textbf{H}^\mathrm{T} \left(  \frac{2}{\sigma^2} \textbf{H} \textbf{K}_\textbf{x}\textbf{H}^\mathrm{T} + \textbf{I}_{N_{\rm r}} \right)^{-1} \textbf{H} \left({\bf{s}}_i-{\bf{s}}_j \right).
\end{aligned}
\right.
\end{equation} 
Consequently, $\log_2 \left[ \mathbb{E}_{\mathcal{Y}} \left( \frac{p_{\mathcal{Y}}}{p_{\mathcal{Y}^C}} \right) \right]$ is given by: 
\begin{equation}
\begin{split}
\log_2 \left[ \mathbb{E}_{\mathcal{Y}} \left( \frac{p_{\mathcal{Y}}}{p_{\mathcal{Y}^C}} \right) \right] &= \frac{1}{2} \log_2 \left[ \frac{| \frac{1}{\sigma^2} \textbf{H} \textbf{K}_\textbf{x}\textbf{H}^\mathrm{T} + \textbf{I}_{N_{\rm r}}|}{| \frac{2}{\sigma^2} \textbf{H} \textbf{K}_\textbf{x}\textbf{H}^\mathrm{T} + \textbf{I}_{N_{\rm r}}|} \right] + \frac{1}{2}  \log_2 \left[\left(2 \pi e \right)^{N_{\rm r}} |\textbf{K}_c| \right] - \frac{N_{\rm r}}{2} \log_2 \left[2 \pi e \sigma^2 \right] \\ 
&+ \log_2 \left[ \sum_{i=1}^{K} \sum_{j=1}^{K} \frac{1}{K^2} \exp \left[ d_{i,j} \right] \right]. 
\end{split}
\end{equation}
Finally, since $ h(\textbf{y}^C)  = \frac{1}{2} \log_2 \left[\left(2 \pi e \right)^{N_{\rm r}} |\textbf{K}_c| \right]$, $h(\textbf{y})$ can be lower bounded as:\vspace{-0.2cm}
\begin{equation}
\begin{split}
h(\textbf{y}) \geq 2 \log_2 (K) + \frac{1}{2} \log_2 \left[ \frac{|  \frac{2}{\sigma^2} \textbf{H} \textbf{K}_\textbf{x}\textbf{H}^\mathrm{T} + \textbf{I}_{N_{\rm r}}|}{| \frac{1}{\sigma^2} \textbf{H} \textbf{K}_\textbf{x}\textbf{H}^\mathrm{T} + \textbf{I}_{N_{\rm r}}|} \right] - \log_2 \left[ \sum_{i=1}^{K} \sum_{j=1}^{K}  \exp \left[ d_{i,j} \right] \right] + \frac{N_{\rm r}}{2} \log_2 \left[2 \pi e \sigma^2 \right].
\end{split}
\end{equation} 
Consequently, by substituting each term in (22) by its expression, the mutual information $I(\textbf{x};\textbf{y})$ can be lower bounded as $ L_2^+ \leq I(\textbf{x};\textbf{y})$, where 
\begin{equation}
\begin{split}
L_2 = 2 \log_2 (K) + \frac{1}{2} \log_2 \left[ \frac{|  \frac{2}{\sigma^2} \textbf{H} \textbf{K}_\textbf{x}\textbf{H}^\mathrm{T} + \textbf{I}_{N_{\rm r}}|}{| \frac{1}{\sigma^2} \textbf{H} \textbf{K}_\textbf{x}\textbf{H}^\mathrm{T} + \textbf{I}_{N_{\rm r}}|} \right] - \log_2 \left[ \sum_{i=1}^{K} \sum_{j=1}^{K}  \exp \left[ d_{i,j} \right] \right].
\end{split}
\end{equation}
The covariance matrix $\textbf{K}_\textbf{x}$ of the transmitted signal $\textbf{x}$ is derived assuming that the transmitted symbols from different light sources are uncorrelated. Each light sources transmits an $M$ PAM symbol with probability $\frac{1}{K}$ and $0$ with probability $1-1/N_{\rm t}$. Therefore, $\textbf{K}_\textbf{x} = \sigma_x^2 \textbf{I}_{N_{\rm t}}$, where $\sigma_x^2 = \frac{1}{K} \sum_{i=1}^m \left(I_m - I\right)^2 = \frac{I^2}{3 N_{\rm t}} \frac{M-1}{M+1}$, which completes the proof.
\bibliographystyle{IEEEtranTCOM}
\bibliography{refs}
\end{document}